\documentclass[11pt,letterpaper]{article}

\usepackage[T1]{fontenc}
\usepackage[margin=1in]{geometry}
\usepackage{newpxtext}
\usepackage{amsmath,amsthm}
\usepackage{newpxmath}

\usepackage[hyphens]{url}
\usepackage[dvipsnames]{xcolor}
\usepackage{tikz}
\usetikzlibrary{arrows.meta,decorations.pathreplacing,calc}
\usepackage[numbers,sort&compress]{natbib}
\usepackage{caption}
\usepackage{wrapfig}
\usepackage{booktabs}
\usepackage[ruled,vlined]{algorithm2e}
\usepackage{xspace}
\usepackage{aliascnt}
\usepackage[
  colorlinks=true,
  linkcolor=blue,
  citecolor=teal,
  urlcolor=blue
]{hyperref}
\usepackage[capitalise,noabbrev,nameinlink]{cleveref}

\crefname{appendix}{appendix}{appendices}
\Crefname{appendix}{Appendix}{Appendices}

\allowdisplaybreaks

\newtheorem{theorem}{Theorem}[section]
\newaliascnt{lemma}{theorem}
\newtheorem{lemma}[lemma]{Lemma}
\aliascntresetthe{lemma}
\newaliascnt{proposition}{theorem}
\newtheorem{proposition}[proposition]{Proposition}
\aliascntresetthe{proposition}
\newaliascnt{corollary}{theorem}
\newtheorem{corollary}[corollary]{Corollary}
\aliascntresetthe{corollary}
\newtheorem*{restatedmaintheorem}{\Cref{thm:main} (restated)}

\crefname{theorem}{theorem}{theorems}
\Crefname{theorem}{Theorem}{Theorems}
\crefname{lemma}{lemma}{lemmas}
\Crefname{lemma}{Lemma}{Lemmas}
\crefname{proposition}{proposition}{propositions}
\Crefname{proposition}{Proposition}{Propositions}
\crefname{corollary}{corollary}{corollaries}
\Crefname{corollary}{Corollary}{Corollaries}

\newcommand{\Main}{\textnormal{\textsc{Main}}\xspace}
\newcommand{\Core}{\textnormal{\textsc{Core}}\xspace}
\newcommand{\SubCore}{\textnormal{\textsc{SubCore}}\xspace}
\newcommand{\Eval}{\operatorname{Eval}}
\newcommand{\Cut}{\operatorname{Cut}}

\crefname{algocf}{algorithm}{algorithms}
\Crefname{algocf}{Algorithm}{Algorithms}
\SetKw{Return}{return}
\definecolor{StageCommentGreen}{RGB}{0,100,0}

\SetCommentSty{AlgoComment}
\newcommand{\AlgoLine}{\csname @endalgoln\endcsname}

\tikzset{
  agent/.style={
    circle,
    draw,
    thick,
    minimum size=8mm,
    inner sep=0pt,
    font=\small
  }
}

\title{Cutting Down the Tower:\\Single-Exponential Envy-Free Cake Cutting}
\author{
    Qilin Ye$^\dagger$\\
	Stanford University\\
	\texttt{yql@stanford.edu}
    \and
    Yannan Bai$^\dagger$\\
	Carnegie Mellon University\\
	\texttt{byn@cmu.edu}
}
\date{}

\begin{document}
\hypersetup{pageanchor=false}
\pagenumbering{gobble}
\maketitle
{\renewcommand{\thefootnote}{\fnsymbol{footnote}}\footnotetext[2]{Equal contribution. Yannan Bai has graciously agreed to put Qilin Ye's name in front, granting someone with the last name Ye the (almost) once-in-a-lifetime privilege of appearing first on a theoretical computer science paper.}}

\begin{abstract}
Envy-free cake cutting is a central problem in fair division with a striking divide between existence and computation. Classical topology guarantees that envy-free allocations exist, yet finding one efficiently (in the Robertson--Webb query model \cite{robertson-webb1998}) has resisted decades of work. The seminal work by Aziz and Mackenzie \cite{aziz-mackenzie} established the existence of a bounded protocol for every $n$, but its query bound is $n^{n^{n^{n^{n^n}}}}$. Subsequently, Sokolov \cite{sokolov} modified their protocol and sharpened its analysis, obtaining a bound of $n^{8n^2(1+o(1))}$ queries, the best deterministic upper bound known before this work. 
In contrast, the general lower bound, due to Procaccia \cite{procaccia2009}, is merely $\Omega(n^2)$.

\vspace{10pt}

We close much of this gap with a deterministic protocol using $n^{O(1)}2^n$ queries, the first single-exponential bound for general valuations. Like Aziz and Mackenzie, we repeatedly allocate cake without envy until some agents can leave with their pieces, guaranteed not to envy anyone who stays, even if a staying agent receives all the remaining cake. To reach this state, Aziz and Mackenzie generate $n^{n^{n^n}}$ partial allocations so that enough of them can be redistributed without conflict. Instead, we show that only polynomially many partial allocations suffice. Our key idea is to ``prune'' scarce allocations upfront: if only a few allocations can help a certain pair of agents, we discard them and handle that pair another way, leaving plenty of allocations for every remaining pair so each gets a distinct one. As a result, we make only polynomially many calls to Aziz and Mackenzie's subroutines, and this gives us our $n^{O(1)}2^n$ bound.

\end{abstract}

\clearpage

\begingroup
\hypersetup{linkcolor=black}
\makeatletter
\let\tocsectionline\l@section
\renewcommand*\l@section[2]{\begingroup\hypersetup{linkcolor=purple}\tocsectionline{#1}{\color{purple}#2}\endgroup}
\let\tocsubsectionline\l@subsection
\renewcommand*\l@subsection[2]{\tocsubsectionline{#1}{\color{purple}#2}}
\makeatother
\tableofcontents
\endgroup
\clearpage

\pagenumbering{arabic}
\hypersetup{pageanchor=true}
\section{Introduction}\label{sec:introduction}

How should we divide a heterogeneous resource among $n$ agents who value its parts differently?
The cake cutting problem captures this question by modeling the resource as an interval and asking for a fair partition.
Since Steinhaus formulated the problem in 1948, cake cutting has become a standard language for fair division in mathematics, economics, and computer science \cite{steinhaus1948,robertson-webb1998}.
In this paper, we focus on the fairness requirement of cake cutting: an allocation is \emph{envy-free} if every agent values her own piece at least as much as every other agent's piece.

Envy-free allocations have long been known to exist, and one short and elegant proof uses Sperner's lemma \cite{stromquist1980,su1999}.
Finding one, however, is much harder.
Since a protocol cannot read an arbitrary valuation in full, the standard \textit{Robertson--Webb model} \cite{robertson-webb1998} limits it to two kinds of queries---an evaluation query asking for the value of an interval, and a cut query for a point that cuts off a prescribed value.

We say a protocol is \emph{finite} if it terminates for every collection of valuations, even when the number of queries depends on those valuations, and it is \emph{bounded} if a function of $n$ alone bounds that number.

\subsection{Previous Work}\label{subsec:previous-results}

For arbitrary nonatomic additive valuations, decades of work have produced only a few general results.
For two agents, cut-and-choose solves the problem immediately.\footnote{One agent divides the cake into two pieces she values equally, the other agent chooses the piece she prefers, and the first agent receives the remaining piece, so neither agent envies the other.}
In the 1960s, Selfridge and Conway independently discovered a bounded envy-free protocol for three agents \cite{brams-taylor1996,procaccia2009}.

For a larger number of agents, progress was much slower.
In 1995, Brams and Taylor gave the first finite discrete envy-free protocol for any number of agents \cite{brams-taylor1995}.
Their protocol always terminates, but its query count is not bounded by any function of $n$.
Other finite protocols had the same limitation \cite{robertson-webb1998,pikhurko2000}.
On the lower-bound side, Procaccia proved in 2009 that every envy-free protocol requires $\Omega(n^2)$ queries \cite{procaccia2009}.
The main open question was whether a bounded protocol existed at all.

Aziz and Mackenzie answered this question affirmatively in 2016 by designing a bounded protocol for four agents and then extending it to every $n$ \cite{aziz-mackenzie-four,aziz-mackenzie}.
Their protocol, which we refer to throughout the paper as the \emph{AM protocol}, uses $n\uparrow\uparrow 6$ Robertson--Webb queries.\footnote{Here $n\uparrow\uparrow k$ denotes a tower of $k$ copies of $n$, evaluated from the top. For example, $n\uparrow\uparrow 2 = n^n$, and $n\uparrow\uparrow 3 = n^{n^n}$.}
Our work builds on two of its main subroutines, \Core and \SubCore, both of which we introduce in \Cref{sec:core-subroutines}.\footnote{Throughout, we refer to AM's final arXiv version, \url{https://arxiv.org/abs/1604.03655v12}.}
Before our work, the only known general improvement to AM's upper bound was due to Sokolov, who modified the AM protocol and sharpened its analysis, especially of \Core and \SubCore, to obtain $n^{8n^2(1+o(1))}$ Robertson--Webb queries \cite{sokolov}.

We discuss additional related work in \Cref{app:related-work}.

\subsection{Our Results and Techniques}

In this paper, we reduce the general upper bound to a single exponential.
More precisely, we prove the following.

\begin{theorem}[Main Theorem]\label{thm:main}
There is a deterministic Robertson--Webb protocol that, for any $n$ agents with nonnegative, additive, nonatomic valuations, returns a complete envy-free allocation using at most $n^{O(1)}2^n$ queries.
\end{theorem}

We summarize our improvement alongside the previously known general bounds in \Cref{tab:query-comparison} and briefly discuss our technical novelty below.

\begin{table}[!ht]
	\centering
	\begin{tabular}{ccc}
		\toprule
		Result & Setting & Robertson--Webb Queries \\
		\midrule
		Procaccia (2009) \cite{procaccia2009} & Lower Bound & $\Omega(n^2)$ \\
		Aziz--Mackenzie (2016) \cite{aziz-mackenzie} & Upper Bound & $n\uparrow\uparrow 6$ \\
		Sokolov (2023) \cite{sokolov} & Upper Bound & $n^{8n^2(1+o(1))} = 2^{(8+o(1))n^2\log_2 n}$ \\
		This paper & Upper Bound & $n^{O(1)}2^n$ \\
		\bottomrule
	\end{tabular}
	\caption{A table of previously known deterministic protocols.}
	\label{tab:query-comparison}
\end{table}

Like Aziz and Mackenzie, our protocol runs recursively. At a high level, we allocate some cake to the agents without creating envy and call the unallocated cake the \emph{residue}. 
Our goal is to split the agents into two nonempty groups---those who leave and those who stay---so that the agents who leave \emph{dominate} those who stay: no departing agent would envy any remaining agent, even if that remaining agent received the \textit{entire} residue on top of her current piece. 
\textit{Once this holds, the departing agents can leave with their pieces, and the remaining agents can recursively divide the residue.} Thus, each recursive step comes down to establishing domination efficiently.

Aziz and Mackenzie's protocol may produce $n\uparrow\uparrow 4$ partial allocations before establishing domination and recursing on fewer agents.
Replacing this loop with polynomially many allocations, however, requires overcoming two challenges.

\paragraph{Challenge \#1.} First, suppose Alice values her own cake piece only slightly more than another piece $P$.
Since her valuation is essentially unrestricted, this margin may be arbitrarily small.
To establish domination, AM must make the residue worth no more to Alice than this margin, and reducing the residue to an arbitrarily small margin can be costly.
AM achieves the reduction via repeatedly calling its \Core subroutine with Alice as the \emph{cutter}.
In each call, Alice divides the residue into pieces of equal value to her.
\Core gives two of these pieces in full, one to Alice and one to another agent, while the unallocated cake becomes the new residue.

As Alice cuts the cake, \Core guarantees that her value for the residue is reduced by a fixed multiplicative factor. Still, reaching her arbitrarily small target may require arbitrarily many calls.
To avoid these arbitrarily costly calls, AM uses an enormous fixed construction.

We improve AM's procedure as follows.
If Alice's margin over $P$ is ``tiny,'' we add enough cake to $P$ so that Alice is now willing to take it instead of her original piece.

We also ensure that this change is insignificant to every other agent, so everyone who already strongly prefers her own piece continues to do so.
If, instead, the margin is ``huge,'' we use a polynomial number of \Core calls to make the residue small compared with it.
We later show that only polynomially many such margins arise, and each requires this treatment at most once, so the entire process uses only polynomially many allocations and subroutine calls.
(We also ensure that every such margin is either ``tiny'' or ``huge.'')
This resolves the first challenge.

\paragraph{Challenge \#2.} To handle a tiny margin from the previous challenge, Alice may need a \Core allocation distributed in a particular way, say, with its piece \#42 going to her.
However, to handle another margin, Bob may need that same allocation distributed differently, with the same piece going to Bob.
Since one piece can go to only one agent, this allocation cannot meet both needs.
To address this, AM's construction involves an enormous family of \Core allocations, so that it carries out the two assignments in different allocations, and no actual cake piece is assigned twice.

We instead begin with only polynomially many \Core allocations and ``prune'' the ones that endanger us. Roughly speaking, we want the remaining allocations to be safe to use, in the sense that using one allocation to meet Alice's need does not prevent us from finding another to meet Bob's. Our key insight is that, if only a few allocations can meet some need, then they are too dangerous to rely on, because other assignments may use them up. Consequently, \textit{we discard these allocations upfront} and later meet this need in another way. With careful bookkeeping, we show that after pruning, every need still handled in this way can be met using many of the remaining allocations. Therefore, we can use a distinct allocation for each need, thereby preventing all conflicts. With only polynomially many allocations, we achieve the same effect as AM's enormous family: we meet Alice's need in one allocation and Bob's in another, without assigning any cake twice.

\vspace{5pt}

With these challenges resolved, we trace the recursion and obtain the $n^{O(1)}2^n$ Robertson--Webb query bound, as claimed.

\paragraph{Subsequent work.}
Following the initial posting (v1) of our work, Mackenzie and Suzuki \cite{mackenzie-suzuki} gave a randomized protocol that terminates almost surely, returns a complete and exactly envy-free allocation, and uses $\widetilde O(n^{30})$ Robertson--Webb queries in expectation over its random choices, for every valuation profile.
Their protocol keeps our framework of combining polynomially many partial allocations, while replacing \Core and \SubCore, which cause the exponential cost of our protocol (see \Cref{sec:discussion}), with new randomized procedures.
Separately, Psomas \cite{psomas} gave a deterministic protocol that uses $2^{O(n\log n)}$ queries, with a short and elementary proof.
At the time of writing, our bound remains the best known among deterministic protocols.

In this version (v2), we streamline several proofs, reorganize the presentation, and, with the help of AI (see also our \Cref{app:ai-acknowledgment}), repair a gap in how v1 handled ties in \SubCore (\Cref{app:subcore}).

\section{Preliminaries}\label{sec:preliminaries}\label{sec:rw-model}\label{sec:domination}

We model the cake as the interval $[0,1]$ and let $\mathcal A$ be the finite set of agents.
We assume each agent $i\in\mathcal A$ has a nonnegative, additive, nonatomic \emph{valuation} $V_i$, and we allow a \emph{piece} (of cake) to be a finite union of intervals.
In other words, $V_i(A\cup B)=V_i(A)+V_i(B)$ for disjoint pieces $A$ and $B$ by additivity, and we may cut any interval to any intermediate value by nonatomicity.
We call a piece that is to be divided on its own a \emph{subcake}.
For a set of agents $N\subseteq\mathcal A$ and a subcake $C$, an allocation $X=(X_i)_{i\in N}$ consists of pairwise disjoint pieces of $C$, and its \emph{residue} is the unallocated part $R=C\backslash\bigcup_{i\in N}X_i$.
The allocation is \emph{complete} when the residue is empty (up to cut endpoints), and it is \emph{envy-free} when $V_i(X_i)\ge V_i(X_j)$ for every $i,j\in N$.

In the \textit{Robertson--Webb} model, a protocol accesses the valuations through two query types:
\begin{itemize}
\item $\Eval_i(x,y)$ returns $V_i([x,y])$.
\item $\Cut_i(x,a)$, for $0\le a\le V_i([x,1])$, returns a point $y\ge x$ satisfying $V_i([x,y])=a$.
\end{itemize}
We say a protocol is \emph{bounded} if some function of the number of agents bounds its number of queries for every collection of valuations. 
Also, we note that the operations on disconnected pieces used later can be simulated by ordinary Robertson--Webb queries.
We give this simulation and explain how we treat exact real-valued oracle answers in \Cref{app:rw-model-details}.

For a partial allocation $X$ with residue $R$, we say that agent $i$ \emph{dominates} agent $j$ if
\[
  V_i(X_i)-V_i(X_j)\ge V_i(R),
\]
i.e., $i$ would not envy $j$ even if $j$ receives all of $R$.

More generally, if a partial allocation $X$ with residue $R$ is envy-free and every agent in $N\backslash E$ dominates every agent in a nonempty set $E\subseteq N$, then any envy-free allocation of $R$ within $E$ can be added to $X$ without making an agent outside $E$ envy an agent in $E$.
We use this observation formally in \Cref{lem:allocation-facts} and \Cref{cor:recursive-reductions}.

\section{The Protocol at a High Level}\label{sec:main-theorem}

Our protocol recursively reduces the number of \textit{active} agents---ones who still need to divide the cake.
A recursive call $\Main(N,C)$ returns a complete envy-free allocation of the (sub)cake $C$ among the active agents $N$.
The top-level call is $\Main(\mathcal A,[0,1])$, and calls with $\le 3$ active agents constitute base cases.
For a larger set $N$ of $n=\lvert N\rvert$ active agents, we seek an envy-free partial allocation with residue $R$ and a nonempty set $E\subsetneq N$, such that every agent in $N\backslash E$ dominates every member of $E$.
We let agents in $N\backslash E$ keep their pieces and recurse on $(E,R)$, which involves fewer active agents.

In this section, we provide a high-level walkthrough of our protocol.
We describe and analyze it in full in \Cref{sec:proof-roadmap} and give the complete pseudocode in \Cref{app:protocol-details,app:subcore}.

Throughout, we repeatedly use two subroutines of Aziz and Mackenzie \cite{aziz-mackenzie}.\footnote{To be precise, we use a variant of \SubCore that satisfies the same guarantees---both on its own and inside \Core---but has a better query bound and handles ties explicitly; see \Cref{prop:subroutine-query-bounds} and \Cref{app:subcore}.}\label{sec:core-subroutines}
Apart from their query costs (\Cref{prop:subroutine-query-bounds}), we rely only on the following properties.
\begin{itemize}
  \item \Core \textbf{shrinks the residue for one agent.}
  Its input consists of the $n$ active agents, the current residue, and one of these agents as the \emph{cutter}, and it returns an envy-free partial allocation of the residue, in which the cutter and at least one other agent each receive a piece that the cutter values at $1/n$ of the residue.
  Thus, the cutter's value for the remaining residue is at most $(n-2)/n$ of what it was.
  Since we can compound these reductions, when we need to reduce Alice's value by a larger factor, we can use a \emph{prescribed sequence} of repeated \Core calls on successive residues with the same set of agents and Alice as the cutter.
  \item \SubCore \textbf{redistributes cake within a group.}
  We use it when we have given some cake to each agent in a group, but the agents in that group may still envy one another.
  Its input consists of the group, one piece for each agent, as well as one empty piece.\footnote{\SubCore requires more pieces than agents.}
  We choose an order on the cake in each piece, so that every piece has a beginning and an end.\footnote{A piece may be disconnected, so an order lists its interval components, for instance from left to right; see \Cref{app:rw-model-details}.}
  A \textit{suffix} of a piece is what remains after cutting off its beginning, and \SubCore works by asking agents to \emph{trim} pieces down to suffixes.
  In the end, each agent receives a suffix of a different piece, each prefix (the part being cut off) returns to the residue, and the result is envy-free within the group.
  Moreover, every suffix an agent trims a piece down to, as well as the one she receives, is worth at least as much to her as her least favorite among the group's pieces.
\end{itemize}
We give the formal guarantees in \Cref{sec:core-subroutines-details}.

\subsection{One Recursive Call}\label{sec:one-recursive-call}

Fix a call $\Main(N,C)$.
We construct the desired partial allocation in three stages.

\paragraph{Stage 1.}
Recall from Challenge \#1 that the margin by which an agent prefers her own piece to another piece may be arbitrarily small. In Stage 1, we ensure that every such margin is either ``small'' or ``large,'' so that we can handle it later as promised. We call this work \textit{pre-processing}.

Pre-processing begins by fixing an agent and making polynomially many consecutive \Core calls with her as the cutter.
We store each envy-free allocation returned by the \Core calls as a \textit{snapshot}.
In snapshot $s$, we call the cake allocated to agent $k$ the \textbf{snapshot piece} $c_{s,k}$.
Since we may later give this piece to another agent, we also call it the \textbf{label-$k$ piece} in snapshot $s$.
We define agent $i$'s \emph{margin} over label $k$ in snapshot $s$ as her value for the label-$i$ piece minus her value for the label-$k$ piece, i.e., $V_i(c_{s,i})-V_i(c_{s,k})$.
Because each snapshot is envy-free, every margin is nonnegative.

Compared with agent $i$'s value for the current residue, we categorize her margins as \emph{small}, \emph{large}, or \emph{intermediate}.
To eliminate the intermediate margins, whenever agent $i$ has one, we run a prescribed sequence of \Core calls with $i$ as the cutter.
This shrinks the residue, so the margin becomes large in comparison, and it remains large as the residue shrinks further.
Since there are only polynomially many margins, we repeat this process at most polynomially many times.
At the end, every margin is either small or large.

Let $R^*$ be the residue after Stage 1.
If some agent values $R^*$ at zero, she already dominates everyone else, so we let her keep her pieces and recurse on the remaining agents, skipping the rest of this call.

\paragraph{Stage 2.}
Recall that agent $i$ dominates agent $j$ once $i$ values the residue no more than the amount by which she prefers her own piece to $j$'s, which we call her \textit{advantage} over $j$.
Our goal in Stage 2 is to find, for each pair $(i,j)$ where possible, a way to make this advantage large enough that polynomially many \Core calls can shrink the residue below it, and our idea is to change who receives which snapshot piece.
For example, if we want Alice to have a large advantage over Bob, one way is to give Bob a snapshot piece that Alice finds unappealing compared with hers.

However, a snapshot divides its cake among all agents, so to give $j$ a different piece, we redistribute the whole snapshot, again giving each agent one piece.
We call such a redistribution a \textit{reassignment}.
We seek only reassignments that cost no agent much, since otherwise some agent may come to envy others.
Accordingly, for each snapshot $s$, we say agent $i$ \textit{accepts} label $k$ if she values the label-$k$ piece $c_{s,k}$ ``almost'' as much as her own $c_{s,i}$, and \textit{rejects} label $k$ if she values it substantially less.
Formally, we use the guarantees from Stage 1.
She accepts label $k$ when her margin over it is small, and rejects it when this margin is large, and there is no third case as no intermediate margin exists.

\begin{wrapfigure}{R}{0.43\textwidth}
  \centering
  \begin{tikzpicture}[
    x=.66cm,
    y=.5cm,
    font=\small,
    block/.style={draw=black, thick},
    note/.style={inner sep=0pt, font=\footnotesize, text height=1.4ex, text depth=.3ex},
    valmark/.style={densely dashed, black!65, semithick},
  ]
  \filldraw[block, fill=BurntOrange!22] (0,2.4) rectangle (8.6,3.3);
  \node[note] at (4.3,2.85) {her own piece};
  \draw[valmark] (8.6,3.65) -- (8.6,0);

  \filldraw[block, fill=BurntOrange!22] (0,1.2) rectangle (8.1,2.1);
  \filldraw[block, fill=RoyalBlue!50] (8.1,1.2) rectangle (8.9,2.1);
  \node[note] at (4.05,1.65) {a piece she accepts};

  \filldraw[block, fill=BurntOrange!22] (0,0) rectangle (4.2,.9);
  \filldraw[block, fill=RoyalBlue!50] (4.2,0) rectangle (5,.9);
  \node[note] at (2.1,.45) {a piece she rejects};
  \draw[densely dotted, thick, red] (5,.45) -- (8.6,.45);
  \end{tikzpicture}
  \caption{One snapshot, with lengths showing $i$'s values. With the little cake we add in \textcolor{Blue!50!RoyalBlue}{dark blue}, a piece she accepts is worth at least as much to her as her own, while a piece she rejects is still worth far less.}
  \label{fig:one-addition}
\end{wrapfigure}
Ideally, for every pair $(i,j)$, we would like a reassignment of some snapshot, in which every agent receives a piece she accepts, and $j$ receives a piece that $i$ rejects.
Such a reassignment achieves our goal for the pair $(i,j)$: $i$ now has a large advantage over $j$, while every other agent loses almost nothing.
In Stage 3, we remove even this ``almost'' by adding a little cake to the reassigned pieces (see \Cref{fig:one-addition}).

To find such reassignments, we look for labels that $i$ rejects but $j$ accepts.
We call snapshot $s$ a \textit{witness} for an ordered pair $(i,j)$ of distinct agents if it has such a label $k$.

Even for pairs that do have witnesses, two obstacles stand in the way of this ideal.
First, reassignments serving different pairs may conflict if they use the same snapshot, since a snapshot can be distributed in only one way.
This means we need to reserve a different witness snapshot for each pair. 
To this end, we prune the ``bad'' snapshots.
If a pair $(i,j)$ has too few witnesses, we discard all of its witness snapshots (we never reassign them, so their pieces stay with their original owners), leaving that pair with no witness, and hence nothing to reserve. We handle such pairs differently in Stage 3.
We repeat this process until every pair that still has a witness can reserve a distinct snapshot, and with careful bookkeeping, we show that enough snapshots remain.

\begin{wrapfigure}{R}{0.43\textwidth}
  \centering
  \begin{tikzpicture}[
    x=2.15cm, y=.8cm,
    cyclearrow/.style={-{Stealth[length=2.2mm]}, very thick, RoyalBlue!75},
    looparrow/.style={-{Stealth[length=2.2mm]}, very thick, BurntOrange!80},
    otherarrow/.style={-{Stealth[length=2mm]}, densely dashed, draw=black!45},
    legendlabel/.style={anchor=west, inner sep=1pt, font=\footnotesize},
  ]
    \node[agent, fill=black!8] (i) at (0.75,2.25) {$i$};
    \node[agent, fill=black!8] (j) at (0.75,0.65) {$j$};
    \node[agent, fill=black!8] (k) at (2.25,2.25) {$k$};
    \node[agent, fill=black!8] (h) at (2.25,0.65) {$h$};

    \draw[cyclearrow] (j) -- (k);
    \draw[cyclearrow] (k) -- (h);
    \draw[cyclearrow] (h) -- (j);
    \draw[looparrow] (i) edge[loop left, looseness=6] (i);
    \draw[otherarrow] (i) -- (j);
    \draw[otherarrow] (h) to[bend left=28] (i);

    \draw[cyclearrow] (0.3,-.35) -- (0.6,-.35);
    \node[legendlabel, text=RoyalBlue] at (0.65,-.35) {the cycle through $j\to k$};
    \draw[looparrow] (0.3,-.95) -- (0.6,-.95);
    \node[legendlabel, text=BurntOrange] at (0.65,-.95) {self-loops outside the cycle};
    \draw[otherarrow] (0.3,-1.55) -- (0.6,-1.55);
    \node[legendlabel, text=black!60] at (0.65,-1.55) {accepted, not in the matching};
  \end{tikzpicture}
  \caption{A perfect matching in $G_s$ that gives $j$ the label-$k$ piece, which $i$ rejects.}
  \label{fig:witness-matching}
\end{wrapfigure}
Second, even in a witness snapshot, a reassignment that gives $j$ the label-$k$ piece while every agent receives a label she accepts may not exist.
To see when this is possible, for each snapshot that remains after pruning, we form a graph $G_s$ with an edge $a\to k$ when $a$ accepts $k$.

A reassignment in which every agent receives a label she accepts is then a perfect matching in $G_s$.
If $G_s$ is strongly connected, we can always find one that pairs $j$ with label $k$. Indeed, the edge $j\to k$ lies on a directed cycle, and since every agent accepts her own label, the edges of this cycle, together with the self-loops outside it, form a perfect matching.
See \Cref{fig:witness-matching}.

Naturally, we would wish that every $G_s$ is strongly connected.
Luckily, even when this fails, there is an easy way out.
Indeed, if some $G_s$ is not strongly connected, we can choose a sink strongly connected component $S$.
No edge leaves $S$, so every agent in $S$ rejects every label outside $S$.
Since all other allocations so far are envy-free, every agent in $S$ already has a large advantage over every agent outside $S$, which is exactly what Stage 2 is after!
Hence, we can run prescribed sequences of \Core calls for the agents in $S$, shrinking the residue until $S$ dominates $N\backslash S$, and recurse on $N\backslash S$, skipping the rest of this call.

Otherwise, every $G_s$ is strongly connected, and with this assumption we move on to Stage 3.

\paragraph{Stage 3.}
In this final stage, our goal is to build an envy-free partial allocation $X$ from which we can reach domination.
Across the remaining snapshots, we say the set of labels an agent rejects forms her \textit{profile}, and agents with equal profiles form a \textit{profile class}.
As two agents with the same profile have no witness in either direction, none of the reassignments we seek gives either of them a large advantage over the other. Therefore, we want to keep them together, letting them either both stay or both leave.

At the end, we want to let the agents in one profile class recurse, while all other agents leave with their pieces (we later show that at least two profile classes remain).
To this end, for agents $i,j$ in different profile classes, we want $i$ to have a large advantage over $j$, so that after we shrink the residue, $i$ dominates $j$.

In Stage 3, we handle a pair $(i,j)$ in one of three cases: (i) $i$ and $j$ belong to different classes and $(i,j)$ has a witness, (ii) they belong to different classes and $(i,j)$ has no witness, and (iii) they belong to the same class.
\vspace{5pt}

\textsc{Case 1.} First consider a pair $(i,j)$ in different classes that has a witness.
Pruning lets us reserve for $(i,j)$ a witness snapshot $s$ that no other pair uses.
We reassign $s$ by a perfect matching in $G_s$, as in Stage 2, so that $j$ receives a label $k$ that $i$ rejects.

Next, we remove the ``almost'' from Stage 2.
For every witness snapshot $s$ we reserve and every label $k$, we cut a small piece $A_{s,k}$ from $R^*$, which we call an \textit{addition}, and attach it to $c_{s,k}$.
We choose each addition so that an agent who accepts $k$ is willing to receive the resulting \textbf{enlarged snapshot piece} $c_{s,k}\cup A_{s,k}$, while an agent who rejects $k$ still prefers her own piece.
This is possible because no margin is intermediate: an addition can exceed every small margin, while staying far below every large one (see \Cref{fig:one-addition}).
Since every agent receives a piece she accepts, no agent loses any value in the reassignments.
In particular, $i$ still prefers her own piece to the one $j$ receives.
Summing over the reserved snapshots, $i$ therefore has a large advantage over $j$.

\begin{wrapfigure}{R}{0.43\textwidth}
  \centering
  \begin{tikzpicture}[
    x=1cm, y=1cm,
    font=\small,
    cls/.style={draw, thick, rounded corners=2pt, minimum width=.9cm, minimum height=.55cm, fill=black!8},
    witness/.style={-{Stealth[length=1.8mm]}, very thick, BurntOrange!80},
    rank/.style={-{Stealth[length=1.8mm, open]}, thick, densely dashed, RoyalBlue!75},
    bar/.style={fill=RoyalBlue!15, draw=RoyalBlue},
    ranklabel/.style={font=\scriptsize, inner sep=1pt, text=RoyalBlue},
  ]
    \node[cls] (D)  at (0,2.6)    {$D$};
    \node[cls] (E)  at (-1.25,1.3) {$E$};
    \node[cls] (E2) at (1.25,1.3)  {$E'$};
    \node[cls] (F)  at (0,0)      {$F$};

    \node[ranklabel, anchor=west] at ($(D.east)+(.1,0)$) {rank 0};
    \draw[bar] ($(E.west)+(-.1,-.09)$) rectangle ++(-.25,.18);
    \node[ranklabel, anchor=east] at ($(E.west)+(-.4,0)$) {rank 1};
    \draw[bar] ($(E2.east)+(.1,-.09)$) rectangle ++(.25,.18);
    \node[ranklabel, anchor=west] at ($(E2.east)+(.4,0)$) {rank 1};
    \draw[bar] ($(F.east)+(.1,-.09)$) rectangle ++(.75,.18);
    \node[ranklabel, anchor=west] at ($(F.east)+(.9,0)$) {rank 3};

    \foreach \hi/\lo/\dx/\dy in {D/E/.045/-.055, D/E2/.045/.055, E/F/.045/.055, E2/F/.045/-.055} {
      \draw[witness, transform canvas={shift={(\dx,\dy)}}] (\hi) -- (\lo);
      \draw[rank, transform canvas={shift={(-\dx,-\dy)}}] (\lo) -- (\hi);
    }
    \draw[witness, transform canvas={xshift=2.2pt}] (D) -- (F);
    \draw[rank, transform canvas={xshift=-2.2pt}] (F) -- (D);
    \draw[witness, transform canvas={yshift=2.2pt}] (E) -- (E2);
    \draw[witness, transform canvas={yshift=-2.2pt}] (E2) -- (E);

    \draw[witness] (-2.3,-.85) -- (-1.75,-.85);
    \node[anchor=west, font=\footnotesize, inner sep=1pt, text=BurntOrange] at (-1.65,-.85) {Case 1 (witness snapshot)};
    \draw[rank] (-2.3,-1.3) -- (-1.75,-1.3);
    \node[anchor=west, font=\footnotesize, inner sep=1pt, text=RoyalBlue] at (-1.65,-1.3) {Case 2 (larger rank piece)};
  \end{tikzpicture}
  \caption{Four profile classes, with pickier classes higher and $E$, $E'$ incomparable.
  An arrow from one class to another means that each member of the first gets a large advantage over each member of the second. Note Case 2 arrows always point up.}
  \label{fig:profile-classes}
\end{wrapfigure}

\vspace{5pt}

\textsc{Case 2.} Now consider a pair $(i,j)$ in different classes that has no witness, so every label rejected by $i$ is also rejected by $j$ in every remaining snapshot.
Then the reassignments have not given $i$ a large advantage over $j$, and the additions may even make $i$ slightly envy $j$.
We therefore use the residue once more to give $i$ a large advantage over $j$.

Specifically, we cut one small piece $P_h$ for every agent $h$, which we call her \emph{rank piece}.
To choose their sizes, we compare profile classes by how ``picky'' their members are. If, in every remaining snapshot, members of a class $D$ reject every label that members of a different class $E$ reject, then we say $D$ is \textit{pickier} than $E$.\footnote{Being pickier is a strict partial order on profile classes. In particular, note that two classes can be \textit{incomparable}, when members of each reject some label that members of the other accept.}
We call the number of classes pickier than a class its \textit{rank}, and we make an agent's rank piece larger when her class has a larger rank.
Since $(i,j)$ has no witness, $j$'s class is pickier than $i$'s, and every class pickier than $j$'s is also pickier than $i$'s.
Thus, $i$ receives a larger rank piece than $j$, which gives her a large advantage over $j$.

This extra cake might seem to make $j$ envy $i$.
However, since $i$ and $j$ belong to different classes, $j$ rejects some label that $i$ accepts, so the pair $(j,i)$ has a witness.
As in Case 1, $j$ then has a large advantage over $i$, and we keep every rank piece small enough that it is worth less to $j$ than this advantage.
\Cref{fig:profile-classes} shows how Cases 1 and 2 together cover every two classes in both directions.

\vspace{5pt}

\textsc{Case 3.} Finally, consider a pair $(i,j)$ within the same class.
After Cases 1 and 2, each agent's share consists of cake cut from $R^*$ (her rank piece and additions) together with her reassigned snapshot pieces, one from each reserved snapshot (see \Cref{fig:share-comparison}).
As the additions may make $i$ slightly envy $j$, we use \SubCore to redistribute the shares within each class.
Recall from the beginning of this section that \SubCore gives each agent a suffix of a different share, worth at least as much to her as her least favorite share in the class.
\begin{figure}[h]
  \centering
  \begin{tikzpicture}[
    x=.066\textwidth,
    y=.5cm,
    font=\small,
    block/.style={draw=black, thick},
    note/.style={inner sep=0pt, font=\footnotesize, text height=1.4ex, text depth=.3ex},
    rowlabel/.style={note, anchor=east, xshift=-2pt},
    sidenote/.style={note, anchor=west},
    gaparrow/.style={{Stealth[length=1.3mm]}-{Stealth[length=1.3mm]}, semithick},
    guide/.style={densely dotted, black!60},
  ]
  \newcommand{\sharerow}[6]{%
    \filldraw[block, fill=BurntOrange!22] (0,#1) rectangle (#2,#1+.9);
    \filldraw[block, fill=RoyalBlue!50] (#2,#1) rectangle (#3,#1+.9);
    \filldraw[block, fill=BurntOrange!22] (#3,#1) rectangle (#4,#1+.9);
    \filldraw[block, fill=RoyalBlue!50] (#4,#1) rectangle (#5,#1+.9);
    \filldraw[block, fill=RoyalBlue!15] (#5,#1) rectangle (#6,#1+.9);
  }
  \draw[semithick, decorate, decoration={brace, amplitude=2.5pt}] (.05,4.65) -- node[above=3pt, note] {a reserved snapshot} (3.9,4.65);

  \sharerow{3.6}{3.25}{3.95}{7.2}{7.9}{9.8}
  \node[note] at (1.625,4.05) {snapshot piece};
  \node[note] at (5.575,4.05) {snapshot piece};
  \node[note] at (8.85,4.05) {rank piece};
  \node[rowlabel] at (0,4.05) {$i$'s share};

  \sharerow{2.4}{3.2}{4}{5.55}{6.25}{8.25}
  \node[note] at (4.775,2.85) {$i$ rejects};
  \node[rowlabel] at (0,2.85) {Case 1: $j$'s share};
  \draw[gaparrow] (8.25,2.85) -- (9.8,2.85);
  \node[sidenote] at (10.15,2.85) {large advantage};

  \sharerow{1.2}{3.2}{4}{7.2}{8}{8.6}
  \node[rowlabel] at (0,1.65) {Case 2: $j$'s share};
  \draw[gaparrow] (8.6,1.65) -- (9.8,1.65);
  \node[sidenote] at (10.15,1.65) {large advantage};

  \sharerow{0}{3.2}{4}{7.2}{8.1}{10}
  \node[rowlabel] at (0,.45) {Case 3: $j$'s share};
  \node[sidenote] at (10.15,.45) {slight envy};

  \draw[guide] (9.8,3.6) -- (9.8,0);
  \end{tikzpicture}
  \caption{How agent $i$ sees the shares before \SubCore, with lengths showing her values and additions in \textcolor{Blue!50!RoyalBlue}{dark blue}.}
  \label{fig:share-comparison}
\end{figure}

We also want the residue to stay inside $R^*$, so that the advantages we have built remain large compared with it.
Since the cake that \SubCore trims off returns to the residue, we need this cake to come only from $R^*$.
To this end, we order each share so that the cake from $R^*$ comes first and the reassigned snapshot pieces come last (see \Cref{fig:suffix-schematic}).
To see why this ordering works, we compare how agent $i$ values a share in her class with how she values her own pieces from the reserved snapshots.
\begin{figure}[h]
  \centering
  \begin{tikzpicture}[
    x=.066\textwidth,
    y=.46cm,
    font=\small,
    guide/.style={densely dotted, black!45},
    measure/.style={{Bar[width=1.6mm]}-{Stealth[length=1.8mm]}, semithick},
    measurelabel/.style={anchor=west, inner sep=1pt, font=\footnotesize},
    swatch/.style={draw=black!50, minimum width=4.5mm, minimum height=2.8mm, inner sep=0pt},
    legendlabel/.style={anchor=west, inner sep=2pt, font=\footnotesize},
  ]
  \fill[RoyalBlue!15] (0,0) rectangle (7,1.1);
  \fill[BurntOrange!22] (7,0) rectangle (10,1.1);
  \draw[thick] (0,0) rectangle (10,1.1);
  \draw[thick] (7,0) -- (7,1.1);

  \draw[line width=1.6pt] (2.6,-.3) -- (2.6,1.45);
  \node[anchor=east, inner sep=1pt, font=\footnotesize] at (2.5,-.6) {trim};

  \draw[semithick, decorate, decoration={brace, amplitude=2.5pt}]
    (.1,1.65) -- node[above=4pt, inner sep=1pt, font=\footnotesize, align=center] {trimmed off,\\returns to the residue} (2.45,1.65);
  \draw[semithick, decorate, decoration={brace, amplitude=2.5pt}]
    (2.75,1.65) -- node[above=4pt, inner sep=1pt, font=\footnotesize] {suffix, given to an agent} (9.9,1.65);

  \draw[measure] (10,-.95) -- (7,-.95);
  \node[measurelabel] at (10.15,-.95) {the snapshot pieces at the end};
  \draw[measure] (10,-1.75) -- (2.6,-1.75);
  \node[measurelabel] at (10.15,-1.75) {the suffix left by the trim};
  \draw[guide] (7,0) -- (7,-.95) (2.6,-.3) -- (2.6,-1.75);

  \node[swatch, fill=RoyalBlue!15] (s1) at (2.6,-3.1) {};
  \node[legendlabel, text=RoyalBlue] at (s1.east) {cake cut from $R^*$};
  \node[swatch, fill=BurntOrange!22] (s2) at (6.2,-3.1) {};
  \node[legendlabel, text=BurntOrange] at (s2.east) {snapshot pieces};
  \end{tikzpicture}
  \caption{A share in agent $i$'s class, with lengths showing $i$'s values.
  Importantly, the trim never reaches the snapshot pieces, so they go whole to the agent who receives the suffix.}
  \label{fig:suffix-schematic}
\end{figure}
\begin{itemize}
  \item Because agents in the same class accept the same labels, the additions make every share in $i$'s class worth \textit{at least} as much to her as her own snapshot pieces, so \SubCore never asks her to trim a share to less than this value.
  \item Because every snapshot is envy-free, she values the snapshot pieces at the end of any share \textit{at most} as much as her own.
\end{itemize}
Hence, every trim falls inside the cake cut from $R^*$, so every suffix that \SubCore allocates contains all the snapshot pieces at the end.
Moreover, agents in the same class have the same rank, so every share in $i$'s class also includes a rank piece of the same size as hers.
Thus, \SubCore gives $i$ a suffix worth at least as much to her as her own snapshot pieces together with such a rank piece, which preserves her large advantage over agents in other classes.

\paragraph{Wrapping up.} We now have the envy-free partial allocation $X$ we need for the recursion.
As promised, we let the agents in one profile class $E\subsetneq N$ recurse, while all other agents leave with their pieces.
To make this safe, for each agent $i\notin E$, we use a prescribed sequence of \Core calls to shrink her value for the remaining residue to at most her advantage over every agent in $E$, so that she dominates all of $E$.
Therefore, we let the agents outside $E$ keep their pieces and recurse on $E$ with the final residue.

\section{Proof of the Main Theorem}\label{sec:proof-roadmap}
We now formalize the ideas from the previous section and use them to prove \Cref{thm:main}.

\subsection{Basic Facts}\label{sec:core-subroutines-details}

We begin with some convenient shorthand notations and facts. For allocations $X$ and $Y$ of disjoint subcakes, we define $X\oplus Y$ by
\[
  (X\oplus Y)_i=X_i\cup Y_i
  \qquad\text{for every agent }i.
\]
If either allocation omits an agent, we treat that agent as receiving the empty piece. We use $\oplus$ to combine several allocations the same way.
For a partial allocation $X$ and agents $i,j$, we denote agent $i$'s advantage over $j$ by
\[
  \Delta_{ij}(X)=V_i(X_i)-V_i(X_j).
\]

Throughout the proof, we need to frequently combine partial allocations, and the following statements come in handy:

\begin{lemma}\label{lem:allocation-facts}
Let $X$ and $Y$ be partial allocations of disjoint subcakes.
\begin{enumerate}
\renewcommand{\labelenumi}{(\roman{enumi})}
\item If $X$ and $Y$ are envy-free among the same set of agents, then $X\oplus Y$ is also envy-free and satisfies $\Delta_{ij}(X\oplus Y)\ge\Delta_{ij}(X)$ for all participating agents $i,j$.
\item If $Y$ is allocated from the residue of $X$ and agent $i$ dominates agent $j$ under $X$, then $i$ still dominates $j$ under $X\oplus Y$.
\end{enumerate}
\end{lemma}

\begin{corollary}\label{cor:recursive-reductions}
Suppose a partial allocation $X$ is envy-free and every agent in $N\backslash E$ dominates every agent in a nonempty set $E\subsetneq N$.
Any complete envy-free allocation of the residue within $E$ can be combined with $X$ to obtain a complete envy-free allocation among all agents in $N$.
\end{corollary}

Recall from \Cref{sec:core-subroutines} that a \emph{prescribed sequence} repeats \Core with the same agents and cutter on successive residues, with a maximum number of calls fixed in advance.
The sequence stops early if the cutter's residue value reaches zero.
We now state the guarantees of Aziz and Mackenzie's \Core \cite{aziz-mackenzie}, then the query bounds for \Core and for our variant of their \SubCore, and finally what we need from this variant.

\begin{proposition}\label{prop:core-guarantees}
Let $N$ be a set of $n$ agents, let $R$ be a residue, and let $c\in N$ be the cutter.
\begin{enumerate}
\renewcommand{\labelenumi}{(\roman{enumi})}
\item A call $\Core(N,R,c)$ returns an envy-free partial allocation $X$ of $R$ with residue $R'$, in which $c$ and at least one other agent each receive a piece worth $V_c(R)/n$ to $c$.
\item The returned residue satisfies
\[
  V_c(R')\le\frac{n-2}{n}V_c(R).
\]
\item For $n\ge3$ and $\beta\ge1$, a prescribed sequence starting from $R$ with maximum length $\left\lceil\log_{n/(n-2)}\beta\right\rceil$ leaves a residue $R'$ satisfying
\[
  V_c(R')\le\frac{V_c(R)}{\beta}.
\]
\end{enumerate}
\end{proposition}

\begin{proposition}\label{prop:subroutine-query-bounds}
Our variant of \SubCore{} (\Cref{alg:subcore-formal}) satisfies the same guarantees as AM's \SubCore, and each call on $m$ agents uses $O(2^m)$ cut queries and $O(m2^m)$ evaluation queries to those agents.\footnote{For the original \SubCore, Sokolov's analysis gives a query bound with base $(3+\sqrt5)/2\approx2.618$ \cite{sokolov}.}
Consequently, a \Core call on $n$ agents uses $O(2^n)$ cut queries and $O(n2^n)$ evaluation queries.
\end{proposition}

\begin{proposition}\label{prop:subcore-required-minima}
Consider a \SubCore call on a set $E$ of agents, with disjoint pieces $W_h$, one for each $h\in E$, and one empty piece.
It returns an envy-free allocation in which every agent receives a suffix of a distinct $W_h$.
Each of these suffixes is either a whole $W_h$ or the suffix left by some agent's trim.
Moreover, every suffix that an agent $i\in E$ trims a piece down to, as well as the one she receives, is worth at least $\min_{h\in E}V_i(W_h)$ to her.
\end{proposition}

We defer the proofs of \Cref{lem:allocation-facts,cor:recursive-reductions} to \Cref{app:omitted-proofs-section-four}, and the proofs of these propositions and the description of our \SubCore variant to \Cref{app:subcore}.

\subsection{Pre-processing the Snapshots}\label{sec:preparation-proof}

We now begin our analysis of the full algorithm, starting with pre-processing as in Stage 1 of \Cref{sec:one-recursive-call}, where we aim to eliminate the intermediate margins.

\phantomsection\label{sec:base-cases}
Fix a call $\Main(N,C)$, and (with slight abuse of notation) set $n=\lvert N\rvert$.
For $n\le3$, we use the classical protocols: when $n=1$, the sole agent takes the subcake; when $n=2$, cut-and-choose; and when $n=3$, Selfridge--Conway \cite{brams-taylor1996}.
For $n\ge4$, we proceed as follows.

Our first step is to construct the snapshots.
Fix an agent $c\in N$, who acts as the cutter for all the snapshots.
Throughout, we use $U$ for the residue we currently have, starting with $U=C$.
To begin pre-processing, we set $Q=n(n-1)$ and make $Q^2$ consecutive \Core calls with $c$ as the cutter, each on the residue left by the one before it.
We store the $Q^2$ returned envy-free allocations in a collection $\mathcal S$ as the snapshots and index them by $s$.
Recall that in snapshot $s$, the snapshot piece $c_{s,k}$ of agent $k$ is also called the label-$k$ piece.

In snapshot $s$, define agent $i$'s margin over label $k$ as
\[
  m_{s,i,k}=V_i(c_{s,i})-V_i(c_{s,k}).
\]

As in Stage 1 of \Cref{sec:one-recursive-call}, we group agent $i$'s margins into small, large, and intermediate ones, compared with her value for the current residue.
To make this precise, we choose two thresholds, which we recompute whenever $U$ shrinks:
\[
  \Gamma_i(U)=\frac{V_i(U)}{Q^2},
  \quad
  \epsilon_i(U)=\frac{\Gamma_i(U)}{4n^2(Q+1)}.
\]
We call a margin \emph{small} when it is at most $\epsilon_i(U)$, \emph{large} when it is at least $\Gamma_i(U)$, and \emph{intermediate} when it lies strictly between them.
To eliminate the intermediate margins, whenever one exists and belongs to agent $i$, we run a prescribed sequence of at most $\left\lceil\log_{n/(n-2)}\bigl(4n^2(Q+1)\bigr)\right\rceil$ \Core calls with $i$ as the cutter.
We then add the returned allocations to a partial allocation $H$, which starts out empty, and replace $U$ with the residue left by the sequence.
We repeat this process until no intermediate margin remains, so that every margin is either small or large, and call the resulting residue $R^*$.
Note that we never redistribute $H$, unlike the snapshots.

\begin{lemma}\label{lem:preparation}
The construction returns a partial allocation $H$, a collection $\mathcal S$ of $Q^2$ snapshots, and a residue $R^*$ such that:
\begin{enumerate}
\renewcommand{\labelenumi}{(\roman{enumi})}
\item $H$ and every snapshot are envy-free, and their allocated pieces together with $R^*$ partition $C$;
\item every margin satisfies
\[
  m_{s,i,k}\le\epsilon_i(R^*)
  \qquad\text{or}\qquad
  m_{s,i,k}\ge\Gamma_i(R^*).
\]
\end{enumerate}
Furthermore, pre-processing uses $n^{O(1)}2^n$ Robertson--Webb queries.
\end{lemma}

\begin{proof}
For (i), by \Cref{prop:core-guarantees} (i), every \Core call returns an envy-free allocation.
Because each call uses the residue left by the preceding one, all these allocations use pairwise disjoint cake, and their allocated pieces, together with $R^*$, partition $C$.
Thus, every snapshot is envy-free, and $H$ is envy-free by \Cref{lem:allocation-facts} (i).

For (ii), we want to show that pre-processing terminates with no intermediate margin.
Suppose $m_{s,i,k}$ is the intermediate margin being processed while the current residue is $U$.
By definition, $\epsilon_i(U)<m_{s,i,k}<\Gamma_i(U)$.
By \Cref{prop:core-guarantees} (iii) with $\beta=4n^2(Q+1)$, the sequence leaves a residue $U'$ satisfying $V_i(U')\le V_i(U)/\beta$.
Thus,
\[
  \Gamma_i(U')
  =\frac{V_i(U')}{Q^2}
  \le\frac{V_i(U)}{4n^2(Q+1)Q^2}
  =\epsilon_i(U)
  <m_{s,i,k}.
\]
Since $m_{s,i,k}$ is fixed while later residues only shrink $\Gamma_i(\cdot)$, the margin $m_{s,i,k}$ becomes and remains large for the rest of pre-processing.

Since $m_{s,i,i}=0$, at most $Q^3$ margins can become intermediate, namely the $Q$ margins with $k\ne i$ in each of the $Q^2$ snapshots.
Each step makes one of them permanently large, so the same margin is never selected again.
The process therefore stops after at most $Q^3$ steps, with no intermediate margin remaining.

Finally, we count the queries.
Constructing the snapshots takes $Q^2$ \Core calls, and evaluating every snapshot piece for every agent takes $O(n^2Q^2)$ queries.
By the argument above, there are at most $Q^3$ later sequences, each containing $O(n\log n)$ \Core calls.
Since $Q=n(n-1)$, we make only polynomially many \Core calls.
Combining this with \Cref{prop:subroutine-query-bounds,prop:ordered-piece-simulation} gives the claimed query bound.
\end{proof}

If some agent $i$ values $R^*$ at zero, then she already dominates every other agent.
To see this, consider the allocation $Y_0$ that combines $H$ with all snapshots in $\mathcal S$, so every agent $j$ receives $H_j\cup\bigcup_{s\in\mathcal S}c_{s,j}$.
By \Cref{lem:preparation,lem:allocation-facts}, $Y_0$ is envy-free, so $\Delta_{ij}(Y_0)\ge0=V_i(R^*)$ for every agent $j$.
Therefore, we let agent $i$ leave with her pieces, while the remaining agents keep theirs and recursively divide $R^*$ among themselves.
That is, we return $Y_0\oplus\Main(N\backslash\{i\},R^*)$, which is complete and envy-free by \Cref{cor:recursive-reductions}.

For the rest of the proof, we assume that every agent values $R^*$ positively.

\subsection{Pruning the Snapshots}\label{sec:reassignments-proof}

With pre-processing complete, we now begin Stage 2.
Unlike the conceptual heavy lifting in \Cref{sec:one-recursive-call}, the math here is surprisingly straightforward: we just need to prune the snapshots so that the reassignments chosen later do not conflict.
For every agent $i$, we abbreviate $\Gamma_i=\Gamma_i(R^*)$ and $\epsilon_i=\epsilon_i(R^*)$.
In snapshot $s$, let $F_{s,i}$ be the set of labels that agent $i$ \emph{rejects}, i.e.,
\[
  F_{s,i}=\{k\in N:m_{s,i,k}>\epsilon_i\},
\]
and say that she \emph{accepts} the labels $k\notin F_{s,i}$.
By \Cref{lem:preparation}, if agent $i$ rejects label $k$ in snapshot $s$, then her margin $m_{s,i,k}$ over that label is large, and if she accepts it, then this margin is small.

Recall that snapshot $s$ is a witness for the pair $(i,j)$ if $i$ rejects some label that $j$ accepts, i.e., $F_{s,i}\nsubseteq F_{s,j}$.
To reserve a different witness snapshot for each pair, we want every pair that has a witness to have at least $Q$ of them, since there are $Q$ ordered pairs of distinct agents.
Therefore, whenever a pair has between $1$ and $Q-1$ witnesses, we discard all of its witness snapshots.
We repeat until no such pair remains and call the remaining collection $\mathcal P$.

\begin{lemma}\label{lem:stable-pool}
The collection $\mathcal P$ contains at least $Q$ snapshots,
and every pair of agents has either no witness or at least $Q$ witnesses in $\mathcal P$.
\end{lemma}

\begin{proof}
Once we discard the witness snapshots of a pair, this pair has no witnesses and cannot regain one as the collection shrinks, so each of the $Q$ pairs causes at most one such step.
Each step discards at most $Q-1$ snapshots.
Consequently, the remaining collection satisfies $\lvert\mathcal P\rvert\ge Q^2-Q(Q-1)=Q$, and every pair has either no witness or at least $Q$ witnesses.
\end{proof}

Next, recall the graph $G_s$ for every snapshot $s\in\mathcal P$, which contains the edge $i\to k$ whenever $i$ accepts $k$, that is, $k\notin F_{s,i}$.
Every self-loop is present because $m_{s,i,i}=0$.
Our main construction requires every $G_s$ to be strongly connected, so we first handle the case where some $G_s$ is not.

\begin{lemma}\label{lem:scc}
Suppose some $G_s$ is not strongly connected, and let $S$ be a sink strongly connected component of $G_s$.
Then we can construct an envy-free partial allocation in which every agent in $S$ dominates every agent in $N\backslash S$.
We can therefore recurse on $N\backslash S$.
\end{lemma}

\begin{proof}
Since $G_s$ is not strongly connected, $S\subsetneq N$, and no edge of $G_s$ leaves $S$.
Thus, for every $i\in S$ and $j\in N\backslash S$, the label $j$ belongs to $F_{s,i}$, so $m_{s,i,j}\ge\Gamma_i$ by \Cref{lem:preparation}.

For each $i\in S$, we run a prescribed sequence of at most $\lceil\log_{n/(n-2)}Q^2\rceil$ \Core calls with $i$ as the cutter on the current residue.
Let $Y$ be the allocation obtained by adding every allocation returned by these calls to the allocation $Y_0$ from \Cref{sec:preparation-proof}.
By \Cref{prop:core-guarantees} (iii), once the sequence for $i$ ends, she values the residue at most $V_i(R^*)/Q^2=\Gamma_i$, and later calls only shrink it.
Since $H$, every snapshot, and every returned allocation are envy-free, \Cref{lem:allocation-facts} (i) shows that $Y$ is envy-free and that $\Delta_{ij}(Y)\ge m_{s,i,j}\ge\Gamma_i$ for every $i\in S$ and $j\in N\backslash S$.

Hence, every agent in $S$ dominates every agent in $N\backslash S$.
By \Cref{cor:recursive-reductions}, we recurse on $N\backslash S$.
\end{proof}

For the rest of the proof, we assume that every $G_s$ is strongly connected.

\subsection{Constructing the Partial Allocation}\label{sec:partial-allocation-proof}

Having pruned the snapshots, we now enter Stage 3.
Recall that agent $i$'s profile $\mathbf F_i=(F_{s,i})_{s\in\mathcal P}$ consists of the labels she rejects in the remaining snapshots, and that agents with the same profile form a profile class.
We write $\mathbf F_i\preceq\mathbf F_j$ when $F_{s,i}\subseteq F_{s,j}$ for every $s\in\mathcal P$, i.e., no snapshot in $\mathcal P$ witnesses $(i,j)$.
We write $\mathbf F_i\prec\mathbf F_j$ when in addition $F_{s,i}\subsetneq F_{s,j}$ for some $s\in\mathcal P$.
In the language of \Cref{sec:one-recursive-call}, $\mathbf F_i\prec\mathbf F_j$ means that $j$'s profile class is pickier than $i$'s.

The rest of this section handles the three cases from \Cref{sec:one-recursive-call}.
Since it is rather lengthy, we first give a brief outline.
We develop one tool each for Cases 1 and 2 (\Cref{lem:matching,lem:rank-properties}), use both tools in one construction, and then handle Case 3 with \SubCore.
We summarize how these steps fit together in \Cref{fig:stage-three-overview}.

\begin{figure}[htbp]
  \centering
  \begin{tikzpicture}[
    x=1cm,
    y=1cm,
    flowbox/.style={
      rectangle,
      draw,
      thick,
      align=center,
      minimum height=1.08cm,
      inner sep=3pt,
      font=\footnotesize
    },
    flowarrow/.style={-{Stealth[length=2mm]},thick},
    flowline/.style={thick,line join=round}
  ]
    \node[flowbox,text width=1.5cm] (pair) at (0,0.75)
      {A pair\\$(i,h)$};

    \node[flowbox,text width=4.5cm] (witness) at (4.2,1.5)
      {Case 1:\\different classes, with witness\\Matchings (\Cref{lem:matching-selection})};
    \node[flowbox,text width=4.5cm] (rank) at (4.2,0)
      {Case 2:\\different classes, no witness\\Ranks (\Cref{lem:rank-properties})};

    \node[flowbox] (combine) at (9.15,0.75)
      {Cut cake for both cases\\(\Cref{lem:joint-pieces})};
    \node[flowbox,text width=4.5cm] (within) at (8.55,-1.5)
      {Case 3: same class\\Redistribute (\Cref{lem:class-settlement})};
    \node[flowbox] (allocation) at (13.5,-1.5)
      {Get partial allocation\\(\Cref{lem:cake-accounting,lem:cross-class})};

    \draw[flowline] (pair.east) -- (1.4,0.75);
    \draw[flowline] (1.4,1.5) -- (1.4,0);
    \draw[flowarrow] (1.4,1.5) -- (witness.west);
    \draw[flowarrow] (1.4,0) -- (rank.west);
    \draw[flowarrow] (pair.south) |- (within.west);

    \draw[flowline] (witness.east) -- (6.95,1.5) -- (6.95,0) -- (rank.east);
    \draw[flowarrow] (6.95,0.75) -- (combine.west);
    \draw[flowarrow] (combine.south) -- (combine.south |- within.north);
    \draw[flowarrow] (within.east) -- (allocation.west);
  \end{tikzpicture}
  \caption{A flowchart for Stage 3.}
  \label{fig:stage-three-overview}
\end{figure}

Before turning to these cases, we verify, as promised, that at least two profile classes remain, since we later recurse on one of them.

\begin{lemma}\label{lem:positive-cutter}
Agent $c$ rejects at least one label in every snapshot, that is, $F_{s,c}\ne\varnothing$ for every $s\in\mathcal S$.
Consequently, at least two profile classes remain.
\end{lemma}

\begin{proof}
Fix a snapshot $s$, and let $U$ and $T$ be the residues immediately before and after its \Core call.
By \Cref{prop:core-guarantees} (i), $c$ and another agent $a$ each receive a piece worth $V_c(U)/n$ to $c$.
In particular, $V_c(c_{s,c})=V_c(U)/n$.
Since the snapshot pieces $(c_{s,k})_{k\in N}$ together with $T$ partition $U$, we have
\[
  V_c(T)
  =\sum_{k\in N}\bigl[V_c(U)/n-V_c(c_{s,k})\bigr]
  =\sum_{k\in N}m_{s,c,k}.
\]
Every $m_{s,c,k}$ is nonnegative by envy-freeness, and the terms for $k=c$ and $k=a$ are zero, so at most $n-2$ terms are positive.

All later operations only shrink the residue, so $R^*\subseteq T$.
It follows that some label $k$ satisfies
\[
  m_{s,c,k}
  \ge\frac{V_c(T)}{n-2}
  \ge\frac{V_c(R^*)}{n-2}
  >\frac{V_c(R^*)}{Q^2}
  =\Gamma_c,
\]
where we used $V_c(R^*)>0$ and $Q^2>n-2$.
Since $\Gamma_c>\epsilon_c$, agent $c$ rejects label $k$.

Now suppose all agents had the same profile, and fix a snapshot $s\in\mathcal P$.
Then, in $s$, they would all reject the same labels.
Agent $i$ never rejects label $i$, so no other agent could reject label $i$ either.
This holds for every $i$, so no agent could reject any label in $s$.
However, we have shown that $c$ rejects at least one label in $s$, a contradiction.
Thus, at least two profile classes remain.
\end{proof}

\paragraph{Case 1: different classes, with witness.}
We begin with a pair $(i,h)$ with $\mathbf F_i\npreceq\mathbf F_h$, where we want $i$ to have a large advantage over $h$.
Then some snapshot $s$ witnesses $(i,h)$, and we choose a label $k\in F_{s,i}\backslash F_{s,h}$.

As in Stage 2 of \Cref{sec:one-recursive-call}, we want to give the label-$k$ piece to $h$.
Since $h$ accepts $k$, the graph $G_s$ contains the edge $h\to k$.
We want to extend this edge to a reassignment in which every agent receives a label she accepts, i.e., a perfect matching in $G_s$.

\begin{lemma}\label{lem:matching}
Every edge $h\to k$ of a strongly connected graph $G_s$ belongs to a perfect matching that pairs each agent with a distinct label and uses only edges of $G_s$.
\end{lemma}

\begin{proof}
If $h=k$, the identity matching works.
Otherwise, strong connectivity gives a simple directed path from $k$ to $h$.
This path together with $h\to k$ forms a directed cycle.
The cycle edges and the self-loops at all vertices outside the cycle form the required perfect matching.
\end{proof}

We now need such a matching for every pair $(i,h)$ with $\mathbf F_i\npreceq\mathbf F_h$, each in a different witness snapshot.

\begin{lemma}\label{lem:matching-selection}
We can reserve a collection $\mathcal J$ of at most $Q$ distinct snapshots and choose a perfect matching $\mu_s$ in $G_s$ for each $s\in\mathcal J$, so that, for every pair $(i,h)$ with $\mathbf F_i\npreceq\mathbf F_h$, some $s\in\mathcal J$ satisfies $\mu_s(h)\in F_{s,i}\backslash F_{s,h}$.
\end{lemma}

\begin{proof}
We go through the pairs one at a time, starting with $\mathcal J$ empty.
By \Cref{lem:stable-pool}, each such pair has at least $Q$ witness snapshots in $\mathcal P$.
Before we reach the pair $(i,h)$, at most $Q-1$ snapshots have been reserved, so at least one of them remains unused.
We add this snapshot $s$ to $\mathcal J$, choose a label $k\in F_{s,i}\backslash F_{s,h}$, and use \Cref{lem:matching} to choose a perfect matching $\mu_s$ containing the edge $h\to k$.
Since there are at most $Q$ pairs, $\lvert\mathcal J\rvert\le Q$.
\end{proof}

In these reassignments, an agent may still lose a small margin.
To remove this loss, we wish to attach an addition $A_{s,k}$ to each snapshot piece $c_{s,k}$ with $s\in\mathcal J$.
We show in \Cref{lem:cross-class} that $i$ then has a large advantage over $h$.

\paragraph{Case 2: different classes, no witness.}
Now let $(i,h)$ be a pair with $\mathbf F_i\prec\mathbf F_h$.
Since no snapshot witnesses $(i,h)$, the reassignments cannot give $i$ a large advantage over $h$.
As in \Cref{sec:one-recursive-call}, we therefore want to give $i$ a larger rank piece than $h$.
To this end, our key tool is the notion of rank.
Let $\mathcal E$ be the set of profile classes, and, overloading the notation, for $E\in\mathcal E$, we write $\mathbf F_E$ for the profile shared by its agents.
For each class $E$, define its \emph{rank} $\ell_E$ as the number of classes $D\in\mathcal E$ with $\mathbf F_E\prec\mathbf F_D$, that is, the number of classes pickier than $E$, and set $\ell_h=\ell_E$ for every $h\in E$.
The following lemma explains the name.

\begin{lemma}\label{lem:rank-properties}
Every agent's rank is at most $n-1$, so $\sum_{h\in N}\ell_h\le Q$.
Moreover, if $\mathbf F_E\prec\mathbf F_D$ for profile classes $E,D$, then $\ell_E\ge\ell_D+1$.
\end{lemma}

\begin{proof}
The classes counted by $\ell_E$ are distinct from $E$, so $0\le\ell_h\le\lvert\mathcal E\rvert-1\le n-1$ for every $h$, and summing over the $n$ agents gives $\sum_{h\in N}\ell_h\le n(n-1)=Q$.
If $\mathbf F_E\prec\mathbf F_D$, then every class counted by $\ell_D$ is also counted by $\ell_E$, because $\prec$ is transitive, and $D$ itself is counted by $\ell_E$ but not by $\ell_D$.
Hence $\ell_E\ge\ell_D+1$.
\end{proof}

To use the ranks, we wish to give each agent a rank piece whose size grows with her rank.
As in Case 1, \Cref{lem:cross-class} shows that $i$ then has a large advantage over $h$.

\paragraph{Cutting the cake for Cases 1 \textit{\&} 2.}
With the matchings and ranks in place, we now cut from $R^*$ the additions of Case 1 and the rank pieces of Case 2, as described in \Cref{sec:one-recursive-call}.
Doing so requires finding many disjoint pieces of $R^*$, each with roughly the right value for every agent, but fortunately, the following lemma shows that we can achieve this efficiently.

\begin{lemma}\label{lem:simultaneous-cut}
Let $R$ be a piece, and let $\lambda_j\ge0$ for $j\in I$ with $\sum_{j\in I}\lambda_j\le1$.
For any $0<\delta<1$, we can deterministically cut disjoint pieces $C_j\subseteq R$, one for each $j\in I$, such that
\[
  \bigl\lvert V_i(C_j)-\lambda_jV_i(R)\bigr\rvert\le\delta V_i(R)
  \qquad\text{for every }i\in N\text{ and }j\in I.
\]
Furthermore, this uses $O\!\left(n^2\delta^{-2}\log(n\lvert I\rvert)\right)$ Robertson--Webb queries.
\end{lemma}

Informally, each agent $i$ cuts $R$ into pieces she values at a small fraction of $V_i(R)$, and combining all agents' cuts divides $R$ into pieces that are small for every agent.
If we independently place each of these small pieces in $C_j$ with probability $\lambda_j$, then Hoeffding's inequality and a union bound show that all the desired inequalities hold with positive probability.
Since we know how every agent values each small piece, we can then find such a placement without further queries.\footnote{Our proof finds it by an exhaustive search over placements. We discuss ways to avoid this search in \Cref{sec:discussion}, though these alternatives concern the computation between queries and do not affect the Robertson--Webb query bound.}
We give the full proof in \Cref{app:omitted-proofs-section-four}.

We now specify the target sizes for the two kinds of pieces we need from $R^*$.
Set $L=1/(nQ^2)$ and $\delta=L/(16n(Q+1))$.
For every $s\in\mathcal J$ and every label $k$, we also set
\[
  r_{s,k}
  =
  \max\left\{
    0,\,
    \frac{m_{s,i,k}}{V_i(R^*)}
    \ \text{ where }\
    i\in N,\ m_{s,i,k}\le\epsilon_i
  \right\},
\]
i.e., the largest fraction of $R^*$ needed to cover any small margin over label $k$ in snapshot $s$.
We then ask for:
\begin{itemize}
  \item Case 1: the addition $A_{s,k}$ to $c_{s,k}$, of size $r_{s,k}+\delta$.
  \item Case 2: the rank piece $P_h$, of size $L(\ell_h-1/2)+\delta$ (when $\ell_h=0$, we set $P_h=\varnothing$).
\end{itemize}
We feed these requests into \Cref{lem:simultaneous-cut}, using $R^*$ and $\delta$.
To state the resulting bounds, we set $\eta_i=\Gamma_i/[4n(Q+1)]$ for every agent $i$.

\begin{lemma}\label{lem:joint-pieces}
The pieces $A_{s,k}$ and $P_h$ are pairwise disjoint subsets of $R^*$ with the following properties:
\begin{enumerate}
\renewcommand{\labelenumi}{(\roman{enumi})}
\item $V_i(A_{s,k})\ge m_{s,i,k}$, whenever $m_{s,i,k}\le\epsilon_i$;
\item $V_i(A_{s,k})\le\eta_i$, for every $i\in N$, $s\in\mathcal J$, and $k\in N$;
\item $\max\{0,L(\ell_h-1/2)\}V_i(R^*)\le V_i(P_h)\le L\ell_hV_i(R^*)$, for every $i,h\in N$.
\end{enumerate}
Further, they can be produced using polynomially many Robertson--Webb queries.
\end{lemma}

Informally, the sizes we ask for sum to less than $1$, so \Cref{lem:simultaneous-cut} yields these bounds directly.
We give the calculation in \Cref{app:omitted-proofs-section-four}.

Since part (iii) already tells us what the rank pieces are worth, it remains to check that the matchings and additions do what we promised in Case 1.
For every agent $h$, let $M_h$ be the union of the snapshot pieces that the matchings assign to $h$, and let $A_h$ be the union of their additions.
For every agent $i$, let $B_i$ be her total value for her own snapshot pieces, i.e., before any reassignment, so that
\[
  M_h=\bigcup_{s\in\mathcal J}c_{s,\mu_s(h)},
  \qquad
  A_h=\bigcup_{s\in\mathcal J}A_{s,\mu_s(h)},
  \qquad
  B_i=\sum_{s\in\mathcal J}V_i(c_{s,i}).
\]
(Mnemonically, these relate to the \textbf{\underline{m}}atching, \textbf{\underline{a}}dditions, and \textbf{\underline{b}}aseline, respectively.)
We want to compare $B_i$ with $V_i(M_h\cup A_h)$.
Specifically, if $h$ is in $i$'s profile class, then we want this value to be at least $B_i$, so that $i$ loses no value in the reassignments, even after we redistribute cake within her class in Case 3.
If instead $(i,h)$ has a witness, then we want it to be well below $B_i$, so that $i$ has a large advantage over $h$.
To do so, we first work within one snapshot, where we want an agent to be willing to receive the enlarged snapshot piece $c_{s,k}\cup A_{s,k}$ exactly when she accepts $k$.

\begin{lemma}\label{lem:prepared-piece-bounds}
For every $s\in\mathcal J$ and $i,k\in N$,
\[
  \begin{cases}
  V_i(c_{s,k}\cup A_{s,k})\ge V_i(c_{s,i}),
    &k\notin F_{s,i},\\
  V_i(c_{s,k}\cup A_{s,k})
    \le V_i(c_{s,i})-\Gamma_i+\eta_i,
    &k\in F_{s,i}.
  \end{cases}
\]
\end{lemma}

\begin{proof}
If $k\notin F_{s,i}$, then $m_{s,i,k}\le\epsilon_i$, so \Cref{lem:joint-pieces} gives $V_i(A_{s,k})\ge m_{s,i,k}$.
Since $V_i(c_{s,k})=V_i(c_{s,i})-m_{s,i,k}$, the inequality follows.

If $k\in F_{s,i}$, then $m_{s,i,k}>\epsilon_i$.
Since \Cref{lem:preparation} rules out margins strictly between $\epsilon_i$ and $\Gamma_i$, we have $m_{s,i,k}\ge\Gamma_i$.
Together with $V_i(A_{s,k})\le\eta_i$, this gives
\[
  V_i(c_{s,k}\cup A_{s,k})
  \le V_i(c_{s,i})-\Gamma_i+\eta_i. \qedhere
\]
\end{proof}

Summing over the snapshots in $\mathcal J$, we obtain what we want (see \Cref{fig:snapshot-segments}), along with two simple bounds for later use.

\begin{lemma}\label{lem:anchors}
For every $i,h\in N$, we have $V_i(M_h)\le B_i$ and $V_i(A_h)\le Q\eta_i$.
Moreover:
\begin{enumerate}
\renewcommand{\labelenumi}{(\roman{enumi})}
\item if $\mathbf F_i=\mathbf F_h$, then $V_i(M_h\cup A_h)\ge B_i$;
\item if $\mathbf F_i\npreceq\mathbf F_h$, then $V_i(M_h\cup A_h)\le B_i-\Gamma_i+Q\eta_i$.
\end{enumerate}
\end{lemma}

\begin{figure}[h]
  \centering
  \begin{tikzpicture}[
    x=.078\textwidth,
    y=.5cm,
    font=\small,
    rowlabel/.style={anchor=east, inner sep=2pt, font=\footnotesize},
    block/.style={draw=black, thick},
    note/.style={inner sep=1pt, font=\footnotesize},
    gaparrow/.style={{Stealth[length=1.3mm]}-{Stealth[length=1.3mm]}, semithick},
    guide/.style={densely dotted, black!60},
  ]
  \node[anchor=south west, note] at (0,7.85) {$\mathbf F_h=\mathbf F_i$, so $i$ accepts every $\mu_t(h)$};
  \foreach \a/\b in {0/1.9, 2.5/4, 4.6/6.5, 7.1/9.7} {
    \filldraw[block, fill=BurntOrange!22] (\a,6.7) rectangle (\b,7.7);
    \draw[guide] (\b,6.7) -- (\b,5.2);
  }
  \node[rowlabel] at (0,7.2) {own pieces, total $B_i$};
  \foreach \a/\b/\c in {0/1.55/2.15, 2.5/3.75/4.25, 4.6/6.1/6.65, 7.1/9.4/9.95} {
    \filldraw[block, fill=BurntOrange!22] (\a,5.2) rectangle (\b,6.2);
    \filldraw[block, fill=RoyalBlue!50] (\b,5.2) rectangle (\c,6.2);
  }
  \node at (.78,5.7) {$M_h$};
  \node at (1.85,5.7) {$A_h$};
  \node[rowlabel] at (0,5.7) {$M_h\cup A_h$};

  \node[anchor=south west, note] at (0,3.75) {$\mathbf F_i\npreceq\mathbf F_h$, so $i$ rejects $\mu_s(h)$ for some $s$};
  \foreach \a/\b in {0/1.9, 2.5/4, 4.6/6.5, 7.1/9.7} {
    \filldraw[block, fill=BurntOrange!22] (\a,2.6) rectangle (\b,3.6);
    \draw[guide] (\b,2.6) -- (\b,1.1);
  }
  \node at (8.4,3.1) {$s$};
  \node[rowlabel] at (0,3.1) {own pieces};
  \foreach \a/\b/\c in {0/1.6/2.1, 2.5/3.65/4.1, 4.6/6.3/6.8, 7.1/7.7/8.15} {
    \filldraw[block, fill=BurntOrange!22] (\a,1.1) rectangle (\b,2.1);
    \filldraw[block, fill=RoyalBlue!50] (\b,1.1) rectangle (\c,2.1);
  }
  \node[rowlabel] at (0,1.6) {$M_h\cup A_h$};

  \draw[semithick] (6.65,1.1) -- (6.65,.6);
  \node[anchor=north, note] at (6.65,.55) {$\le\eta_i$};
  \draw[gaparrow] (8.15,1.6) -- node[below=2pt, inner sep=1pt, font=\footnotesize] {$\ge\Gamma_i-\eta_i$} (9.7,1.6);
  \end{tikzpicture}
  \caption{How agent $i$ values $M_h\cup A_h$ compared with her own pieces, with one segment for each snapshot in $\mathcal J$ and lengths showing $i$'s values.}
  \label{fig:snapshot-segments}
\end{figure}

\begin{proof}
Because every snapshot in $\mathcal J$ is envy-free, agent $i$ values the snapshot piece assigned to $h$ no more than her own piece in that snapshot.
Summing over the snapshots gives $V_i(M_h)\le B_i$.
By definition, $A_h$ contains one addition from each of the at most $Q$ snapshots in $\mathcal J$.
Each addition is worth at most $\eta_i$ to agent $i$, so $V_i(A_h)\le Q\eta_i$.

For (i), suppose $\mathbf F_i=\mathbf F_h$.
For every $s\in\mathcal J$, the matching uses the edge $h\to\mu_s(h)$, so $\mu_s(h)\notin F_{s,h}=F_{s,i}$.
By \Cref{lem:prepared-piece-bounds},
\[
  V_i\bigl(c_{s,\mu_s(h)}\cup A_{s,\mu_s(h)}\bigr)
  \ge V_i(c_{s,i}).
\]
Summing over all $s\in\mathcal J$ gives $V_i(M_h\cup A_h)\ge B_i$.

For (ii), suppose $\mathbf F_i\npreceq\mathbf F_h$.
By \Cref{lem:matching-selection}, some snapshot $s\in\mathcal J$ satisfies $k=\mu_s(h)\in F_{s,i}$.
For this snapshot, \Cref{lem:prepared-piece-bounds} gives
\[
  V_i(c_{s,k}\cup A_{s,k})
  \le V_i(c_{s,i})-\Gamma_i+\eta_i.
\]
For every $t\in\mathcal J$ other than $s$, envy-freeness of snapshot $t$ and the bound $V_i(A_{t,\mu_t(h)})\le\eta_i$ give
\[
  V_i(c_{t,\mu_t(h)}\cup A_{t,\mu_t(h)})
  \le V_i(c_{t,i})+\eta_i.
\]
Summing these inequalities gives
\[
  V_i(M_h\cup A_h)
  \le B_i-\Gamma_i+\eta_i+(Q-1)\eta_i
  =B_i-\Gamma_i+Q\eta_i. \qedhere
\]
\end{proof}

\paragraph{Case 3: same class.}
It remains to address agents in the same profile class who may still envy one another.
As in \Cref{sec:one-recursive-call}, we plan to use \SubCore here, and we want the cake it trims off to come only from $R^*$.
Therefore, for every agent $h$, we concatenate $P_h$, $A_h$, and $M_h$ in this order (see \Cref{app:rw-model-details}) to form her share
\[
  W_h=[P_h,A_h,M_h]
\]
as in \Cref{fig:class-shares}.
Before we hand these shares to \SubCore, we should make sure that they do not overlap.
Indeed, since each matching $\mu_s$ assigns every label $k$ exactly once, $c_{s,k}$ belongs to exactly one $M_h$, and $A_{s,k}$ to exactly one $A_h$.
Therefore, by \Cref{lem:preparation,lem:joint-pieces}, the shares $W_h$ are pairwise disjoint.

We now fix a profile class $E$ and use \SubCore, with one empty piece as in \Cref{prop:subcore-required-minima}, to redistribute the shares of its agents.
We do not want this redistribution to undo what we have built in Cases 1 and 2, so we want every agent $i\in E$ to still receive a suffix worth at least
\[
  q_i=B_i+\max\{0,L(\ell_E-1/2)\}V_i(R^*),
\]
i.e., her value for her own snapshot pieces plus a lower bound on her value for every $P_h$ with $h\in E$.
We show next that \SubCore guarantees $i$ at least $q_i$, and that every trim of $W_h$ falls inside $P_h\cup A_h$, so that the residue stays inside $R^*$.

\begin{lemma}\label{lem:class-settlement}
\begin{enumerate}
\renewcommand{\labelenumi}{(\roman{enumi})}
\item For every $i,h\in E$, we have $V_i(W_h)\ge q_i$. Consequently, \SubCore returns an envy-free allocation in which every agent $i\in E$ receives a suffix of a distinct $W_h$ worth at least $q_i$.
\item For every $h\in E$, the suffix allocated from $W_h$ contains all of $M_h$. Consequently, the part of $W_h$ left unallocated is a prefix contained in $P_h\cup A_h$.
\end{enumerate}
\end{lemma}

\begin{figure}[h]
  \centering
  \begin{tikzpicture}[
    x=.075\textwidth,
    y=.5cm,
    font=\small,
    rowlabel/.style={anchor=east, inner sep=2pt},
    valmark/.style={densely dashed, black!65, semithick},
    block/.style={draw=black, thick},
    note/.style={inner sep=1pt, font=\footnotesize},
    measure/.style={{Bar[width=1.6mm]}-{Stealth[length=1.8mm]}, semithick},
  ]
  \filldraw[block, fill=RoyalBlue!15] (.6,1.7) rectangle (4.6,2.7);
  \filldraw[block, fill=RoyalBlue!45] (4.6,1.7) rectangle (6.4,2.7);
  \filldraw[block, fill=BurntOrange!22] (6.4,1.7) rectangle (10,2.7);
  \node at (3.8,2.2) {$P_h$};
  \node at (5.05,2.2) {$A_h$};
  \node at (8.2,2.2) {$M_h$};
  \node[rowlabel] at (.6,2.2) {$W_h$};

  \filldraw[block, fill=RoyalBlue!15] (2,.2) rectangle (5,1.2);
  \filldraw[block, fill=RoyalBlue!45] (5,.2) rectangle (6.7,1.2);
  \filldraw[block, fill=BurntOrange!22] (6.7,.2) rectangle (10,1.2);
  \node at (4,.7) {$P_{h'}$};
  \node at (5.35,.7) {$A_{h'}$};
  \node at (8.35,.7) {$M_{h'}$};
  \node[rowlabel] at (2,.7) {$W_{h'}$};

  \draw[valmark] (5.7,.2) -- (5.7,3.9);
  \node[anchor=south, inner sep=1pt] at (5.7,3.9) {$B_i$};
  \draw[valmark] (3,.2) -- (3,3.9);
  \node[anchor=south, inner sep=1pt] at (3,3.9) {$q_i$};
  \draw[semithick, decorate, decoration={brace, amplitude=2.5pt}] (3.05,2.9) -- node[above=3pt, inner sep=1pt, font=\footnotesize] {$L(\ell_E-\frac{1}{2})V_i(R^*)$} (5.65,2.9);

  \draw[densely dotted, thick] (2,-.55) -- (2,1.7);
  \draw[measure] (10,-.55) -- (2,-.55);
  \node[anchor=west, note] at (10.1,-.55) {$\min_{h\in E}V_i(W_h)$};
  \draw[line width=1.6pt] (1.3,1.4) -- (1.3,3.0);
  \node[anchor=south, note] at (1.3,3.0) {trim to value $v$};

  \draw[measure] (10,-1.45) -- (1.3,-1.45);
  \node[anchor=west, note] at (10.1,-1.45) {$v$};
  \draw[densely dotted, black!45] (1.3,1.4) -- (1.3,-1.45);
  \end{tikzpicture}
  \caption{Two shares in $i$'s class, with lengths showing $i$'s values, measured from the right end.}
  \label{fig:class-shares}
\end{figure}

\begin{proof}
For (i), fix agents $i,h\in E$.
Since they have the same profile, \Cref{lem:anchors} gives $V_i(M_h\cup A_h)\ge B_i$.
They also have the same rank, so \Cref{lem:joint-pieces} gives $V_i(P_h)\ge q_i-B_i$.
Adding these two inequalities gives $V_i(W_h)\ge q_i$, and the rest follows from \Cref{prop:subcore-required-minima}.

For (ii), suppose \SubCore asks agent $i$ to trim $W_h$, or one of its suffixes, to value $v$.
We want this trim to leave all of $M_h$ intact.
By \Cref{prop:subcore-required-minima} and (i), $v\ge q_i$, and by \Cref{lem:anchors}, $V_i(M_h)\le B_i\le q_i$.
Thus, $V_i(M_h)\le v$. Since $M_h$ comes last in $W_h$, the suffix left by this trim contains all of $M_h$,\footnote{If $V_i(M_h)=v$, our tiebreakers in \Cref{app:zero-value} keep the trim at the start of $M_h$ or earlier.} as shown in \Cref{fig:class-shares}.
By \Cref{prop:subcore-required-minima}, every allocated suffix is a whole $W_h$ or the suffix left by a trim, so it contains $M_h$.
\end{proof}

\paragraph{The partial allocation.}
We have gathered everything we need, and we can now put the pieces together.
First, to remove any envy within a class, we run \SubCore in every profile class as in \Cref{lem:class-settlement}, and combine the resulting allocations into $Z$.
In the snapshots outside $\mathcal J$, which we never reassign, every agent simply keeps her own snapshot pieces.
We write $\operatorname{Id}(\mathcal S\backslash\mathcal J)$ for this assignment.
Adding $H$ from \Cref{lem:preparation}, we finally define
\[
  X=H\oplus\operatorname{Id}(\mathcal S\backslash\mathcal J)\oplus Z,
\]
the partial allocation we seek and on which we build the recursion.
Before we use $X$, we want to make sure that it is well defined, and that our careful ordering of the shares in Case 3 indeed keeps its residue inside $R^*$.

\begin{lemma}\label{lem:cake-accounting}
The allocation $X$ is well defined, i.e., no cake is allocated twice, and its residue $R_1$ is contained in $R^*$.
\end{lemma}

\begin{proof}
To see that no cake is allocated twice, recall from \Cref{lem:preparation} that the allocated pieces of $H$ and of the snapshots, together with $R^*$, partition $C$.
The allocation $Z$ uses only the shares $W_h$, which are pairwise disjoint and built from the snapshots in $\mathcal J$ and cake in $R^*$.
Since $H$ and $\operatorname{Id}(\mathcal S\backslash\mathcal J)$ use only cake outside $R^*$ and the snapshots in $\mathcal J$, no cake is allocated twice.

As for the residue, by \Cref{lem:class-settlement} (ii), the suffix that $Z$ allocates from every $W_h$ contains all of $M_h$, so $Z$ allocates every snapshot piece from $\mathcal J$.
Since $X$ also allocates all of $H$ and every snapshot outside $\mathcal J$, all the cake it leaves unallocated lies in $R^*$.
\end{proof}

We can now collect what we have been working toward throughout Stage 3.
Recall that we want $X$ to be envy-free, and that we want every agent to have a large advantage over every agent in another profile class, so that, once we shrink the residue, she dominates them.

\begin{lemma}\label{lem:cross-class}
The allocation $X$ is envy-free and, for every pair $(i,j)$ in different profile classes,
\[
  \Delta_{ij}(X)\ge\eta_i
  =\frac{V_i(R^*)}{4n(Q+1)Q^2}.
\]
\end{lemma}

\begin{proof}
We first work with $Z$, and pass to $X$ at the end.
Fix agents $i,j$ in different profile classes.
We want to show that $i$ has a large advantage over $j$ in $Z$, so we compare how she values $Z_i$ and $Z_j$.
By \Cref{lem:class-settlement} (i), we have $V_i(Z_i)\ge q_i$, and agent $j$ receives a suffix $Z_j\subseteq W_h$ for some $h$ in her profile class.
Note that $h$ may not be $j$ herself, but $h$ has the same profile and rank as $j$, i.e., $\mathbf F_h=\mathbf F_j$ and $\ell_h=\ell_j$.
We measure everything in units of $\eta_i$, and by definition, $LV_i(R^*)=4(Q+1)\eta_i$ and $\Gamma_i=4n(Q+1)\eta_i$.

Since their profiles differ, either $\mathbf F_i\npreceq\mathbf F_j$ or $\mathbf F_i\prec\mathbf F_j$.
We visualize both cases in \Cref{fig:cross-class-values}.

\begin{figure}[h]
  \centering
  \begin{tikzpicture}[
    x=.078\textwidth,
    y=.5cm,
    font=\small,
    rowlabel/.style={anchor=east, inner sep=2pt, font=\footnotesize},
    block/.style={draw=black, thick},
    note/.style={inner sep=1pt, font=\footnotesize},
    floorray/.style={-{Stealth[length=2mm]}, very thick},
    gaparrow/.style={{Stealth[length=1.3mm]}-{Stealth[length=1.3mm]}, semithick},
  ]
  \node[anchor=south west, note] at (0,9.75) {Case 1: $\mathbf F_i\npreceq\mathbf F_j$};
  \filldraw[block, fill=BurntOrange!22] (0,8.6) rectangle (7.2,9.6);
  \node at (3.6,9.1) {$B_i$};
  \draw[floorray] (7.2,9.1) -- (7.9,9.1);
  \node[rowlabel] at (0,9.1) {$V_i(Z_i)\ge$};

  \draw[gaparrow] (1.2,8.1) -- node[fill=white, inner sep=1pt, font=\footnotesize] {$\Gamma_i=nLV_i(R^*)$} (7.2,8.1);
  \filldraw[block, fill=BurntOrange!22] (0,6.6) rectangle (1.2,7.6);
  \filldraw[block, fill=RoyalBlue!15] (1.2,6.6) rectangle (5.7,7.6);
  \filldraw[block, fill=RoyalBlue!50] (5.7,6.6) rectangle (6.05,7.6);
  \node[note] at (.6,7.1) {$B_i-\Gamma_i$};
  \node[note] at (3.45,7.1) {$L(n-1)V_i(R^*)$};
  \node[rowlabel] at (0,7.1) {$V_i(Z_j)\le$};
  \node[anchor=north, note] at (5.87,6.55) {$Q\eta_i$};
  \draw[gaparrow] (6.05,7.1) -- (7.2,7.1);
  \node[anchor=west, note] at (7.25,7.1) {$(3Q+4)\eta_i$};

  \node[anchor=south west, note] at (0,4.15) {Case 2: $\mathbf F_i\prec\mathbf F_j$, so $\ell_i\ge\ell_j+1$};
  \filldraw[block, fill=BurntOrange!22] (0,3) rectangle (7.2,4);
  \filldraw[block, fill=RoyalBlue!15] (7.2,3) rectangle (9.45,4);
  \node at (3.6,3.5) {$B_i$};
  \node[note] at (8.32,3.5) {$L(\ell_i-\frac12)V_i(R^*)$};
  \draw[floorray] (9.45,3.5) -- (10.15,3.5);
  \node[anchor=south, note] at (9.45,4.05) {$q_i$};
  \node[rowlabel] at (0,3.5) {$V_i(Z_i)\ge$};

  \filldraw[block, fill=BurntOrange!22] (0,1.5) rectangle (7.2,2.5);
  \filldraw[block, fill=RoyalBlue!15] (7.2,1.5) rectangle (8.7,2.5);
  \filldraw[block, fill=RoyalBlue!50] (8.7,1.5) rectangle (9.05,2.5);
  \node at (3.6,2) {$B_i$};
  \node[note] at (7.95,2) {$L\ell_jV_i(R^*)$};
  \node[rowlabel] at (0,2) {$V_i(Z_j)\le$};
  \node[anchor=north, note] at (8.87,1.45) {$Q\eta_i$};
  \draw[densely dotted, black!55] (9.45,3) -- (9.45,1.5);
  \draw[gaparrow] (9.05,2) -- (9.45,2);
  \node[anchor=west, note] at (9.5,2) {$(Q+2)\eta_i$};

  \draw[densely dashed, black!65, semithick] (7.2,.9) -- (7.2,10.3);
  \node[anchor=south, inner sep=1pt] at (7.2,10.3) {$B_i$};
  \end{tikzpicture}
  \caption{The two cases in the proof of \Cref{lem:cross-class}, with lengths showing $i$'s values.}
  \label{fig:cross-class-values}
\end{figure}

Suppose first $\mathbf F_i\npreceq\mathbf F_j$, i.e., Case 1.
Here, we want the label that $i$ rejects to outweigh $j$'s rank piece.
For $Z_i$, we only need $V_i(Z_i)\ge q_i\ge B_i$.
For $Z_j$, because $\mathbf F_h=\mathbf F_j$, we also have $\mathbf F_i\npreceq\mathbf F_h$, so \Cref{lem:anchors} gives
\[
  V_i(M_h\cup A_h)\le B_i-\Gamma_i+Q\eta_i.
\]
In addition, \Cref{lem:rank-properties,lem:joint-pieces} give $V_i(P_h)\le L(n-1)V_i(R^*)$.
Since $Z_j\subseteq W_h$, adding these two inequalities gives
\[
  V_i(Z_j)\le B_i-\Gamma_i+L(n-1)V_i(R^*)+Q\eta_i.
\]
Together,
\[
\begin{aligned}
  \Delta_{ij}(Z)
  &\ge\Gamma_i-L(n-1)V_i(R^*)-Q\eta_i\\
  &=\bigl[4n(Q+1)-4(n-1)(Q+1)-Q\bigr]\eta_i\\
  &=(3Q+4)\eta_i \ge\eta_i.
\end{aligned}
\]

Now suppose $\mathbf F_i\prec\mathbf F_j$, i.e., Case 2.
Here, we want $i$'s larger rank piece to outweigh $j$'s additions.
By \Cref{lem:rank-properties}, $\ell_i\ge\ell_j+1$, so for $Z_i$, the definition of $q_i$ gives
\[
  V_i(Z_i)
  \ge B_i+L(\ell_i-1/2)V_i(R^*).
\]
For $Z_j$, by \Cref{lem:anchors,lem:joint-pieces}, we have $V_i(M_h)\le B_i$, $V_i(P_h)\le L\ell_jV_i(R^*)$, and $V_i(A_h)\le Q\eta_i$.
Since $Z_j\subseteq W_h$, these inequalities give
\[
  V_i(Z_j)\le B_i+L\ell_jV_i(R^*)+Q\eta_i.
\]
Together,
\[
\begin{aligned}
  \Delta_{ij}(Z)
  &\ge L(\ell_i-\ell_j-1/2)V_i(R^*)-Q\eta_i\ge\frac L2V_i(R^*)-Q\eta_i\\
  &=\bigl[2(Q+1)-Q\bigr]\eta_i =(Q+2)\eta_i\\
  &\ge\eta_i.
\end{aligned}
\]

Since $i$ and $j$ were arbitrary, the two cases give $\Delta_{ij}(Z)\ge\eta_i$ whenever $i$ and $j$ belong to different profile classes, which is the large advantage we want.
It remains to pass from $Z$ to $X$.
By \Cref{lem:class-settlement}, $Z$ is also envy-free within each profile class, so $Z$ is envy-free.
By \Cref{lem:preparation}, $H$ and every snapshot in $\mathcal S\backslash\mathcal J$ are envy-free, and by \Cref{lem:cake-accounting}, these allocations and $Z$ use pairwise disjoint cake.
Thus, \Cref{lem:allocation-facts} (i) gives $\Delta_{ij}(X)\ge\Delta_{ij}(Z)$ for all agents $i,j$, so $X$ is envy-free and the proof is complete.
\end{proof}

\subsection{Recursion and Query Complexity}\label{sec:complexity-proof}

With $X$ in hand, we can now complete the recursion, as we outlined at the end of \Cref{sec:one-recursive-call}, and then count the queries used by the protocol.

Recall that we want the agents in one profile class to recurse, while all other agents leave with their pieces.
Since \Cref{lem:positive-cutter} guarantees that at least two profile classes remain, we can choose a class $E_*$ to recurse on.
To make every agent $i$ outside $E_*$ dominate its members, we want her value for the residue to be at most her advantage over them, which is at least $\eta_i$ by \Cref{lem:cross-class}.
Therefore, for each $i\in N\backslash E_*$, we run a prescribed sequence of at most $\left\lceil\log_{n/(n-2)}\bigl(4n(Q+1)Q^2\bigr)\right\rceil$ \Core calls with $i$ as the cutter, successively on the current residue.
Let $R_2$ be the final residue after all these sequences, and let $X'$ be obtained by adding every allocation returned by these calls to $X$.
We recurse on $E_*$ and return
\[
  X'\oplus\Main(E_*,R_2).
\]

\begin{proposition}\label{prop:assembled-step}
Let $n\ge4$, and suppose every recursive call on fewer than $n$ agents returns a complete envy-free allocation.
Then $\Main(N,C)$ returns a complete envy-free allocation after one recursive call on a nonempty proper subset of agents.
\end{proposition}

\begin{proof}
A call can proceed in three ways, and we have already handled the first two.
If some agent $i$ values $R^*$ at zero, the reduction at the end of \Cref{sec:preparation-proof} makes one recursive call on $N\backslash\{i\}$, and \Cref{cor:recursive-reductions} combines the allocation it returns with $Y_0$.
If instead some $G_s$ is not strongly connected, \Cref{lem:scc} reduces the problem to one recursive call on a nonempty proper subset of agents.

In the remaining case, every $G_s$ is strongly connected, and we follow the construction above.
\Cref{lem:positive-cutter} gives at least two profile classes, so $E_*$ is nonempty and proper.
By \Cref{lem:cross-class,lem:cake-accounting}, the allocation $X$ is envy-free, has residue $R_1\subseteq R^*$, and satisfies $\Delta_{ij}(X)\ge\eta_i$ whenever $i$ and $j$ belong to different profile classes.

We want every agent outside $E_*$ to dominate every member of $E_*$.
To this end, fix $i\in N\backslash E_*$.
When her prescribed sequence begins, the current residue is contained in $R_1$, and all later calls only shrink it.
Hence \Cref{prop:core-guarantees} (iii) gives
\[
  V_i(R_2)
  \le\frac{V_i(R_1)}{4n(Q+1)Q^2}
  \le\frac{V_i(R^*)}{4n(Q+1)Q^2}
  =\eta_i.
\]

Because every added \Core allocation is envy-free, \Cref{lem:allocation-facts} (i) shows that $X'$ remains envy-free and that $\Delta_{ij}(X')\ge\Delta_{ij}(X)$ for every pair of agents.
Consequently, for every $j\in E_*$,
\[
  \Delta_{ij}(X')
  \ge\Delta_{ij}(X)
  \ge\eta_i
  \ge V_i(R_2).
\]

Thus, $i$ dominates every member of $E_*$, as we wanted.

By assumption, $\Main(E_*,R_2)$ returns a complete envy-free allocation of $R_2$, and \Cref{cor:recursive-reductions} lets us combine it with $X'$.
\end{proof}

It remains to count the queries.

\begin{proposition}\label{prop:recursive-step}
There are fixed constants $K,d>0$ such that, for every $n\ge4$, a call $\Main(N,C)$ with $\lvert N\rvert=n$ uses at most $\Lambda_n=Kn^d2^n$ queries, not counting those made by its recursive call.
\end{proposition}

\begin{proof}
The exponential cost comes only from \Core and \SubCore, so the point is that we make only polynomially many calls to them.
Pre-processing uses $n^{O(1)}2^n$ queries by \Cref{lem:preparation}.
After pre-processing, queries are used only for the prescribed \Core sequences, for cutting the additions and rank pieces, and for the \SubCore calls within profile classes.

At most $n$ prescribed sequences are used, each containing $O(n\log n)$ \Core calls.
By \Cref{lem:joint-pieces}, cutting the additions and rank pieces uses polynomially many queries.
We also invoke \SubCore at most once for each profile class.

Consequently, only polynomially many \Core and \SubCore calls are made after pre-processing, each involving at most $n$ agents.
By \Cref{prop:subroutine-query-bounds,prop:ordered-piece-simulation}, these calls use $n^{O(1)}2^n$ queries in total.
Combining these bounds, we can choose fixed constants $K,d>0$ such that $\Lambda_n=Kn^d2^n$ bounds all queries made outside recursion.
\end{proof}

We are finally ready to assemble the main theorem.

\begin{restatedmaintheorem}
There is a deterministic Robertson--Webb protocol that, for any $n$ agents with nonnegative, additive, nonatomic valuations, returns a complete envy-free allocation using at most $n^{O(1)}2^n$ queries.
\end{restatedmaintheorem}

\begin{proof}
We prove correctness by induction on $n$.
The protocols in \Cref{sec:base-cases} handle $n\le3$, and \Cref{prop:assembled-step} gives the inductive step for $n\ge4$.
We make every choice deterministic using the fixed orders and tie-breaking rules described in \Cref{app:deterministic-choices}.

To count the queries, let $T(n)$ be the largest number of queries used by a call with $n$ agents, including its recursive calls.
Since a recursive call may involve any smaller number of agents, we set $\overline T(n)=\max_{1\le m\le n}T(m)$.
By \Cref{prop:assembled-step,prop:recursive-step}, every call uses at most $\Lambda_n$ queries outside recursion and makes at most one recursive call on fewer agents.
Hence,
\[
  \overline T(n)
  \le\overline T(n-1)+\Lambda_n
  \le \overline T(3)+\sum_{k=4}^n\Lambda_k.
\]
Since $\Lambda_k=Kk^d2^k\le Kn^d2^k$ for every $k\le n$, the sum is at most $Kn^d2^{n+1}=2\Lambda_n$.
Therefore, $T(n)\le n^{O(1)}2^n$, as claimed.
\end{proof}

\section{Discussion and Concluding Remarks}\label{sec:discussion}

Our protocol gives the first single-exponential upper bound for finding a complete envy-free allocation.
Yet the $\Omega(n^2)$ lower bound of Procaccia \cite{procaccia2009} leaves a wide gap.
To better understand the complexity of our protocol, we consider two questions.
First, how much computation is needed between queries?
Second and more importantly, where does the exponential dependence in our query bound come from?

For the first question, the main issue appears in \Cref{lem:simultaneous-cut}, where the proof uses an exhaustive search to place the small pieces.
Here, we observe that this search is not necessary.
It turns out that we can instead place each small piece randomly according to the probabilities used in the proof.
The same concentration argument shows that the resulting placement succeeds with high probability.
Because we already know the value of every such piece, we can check whether a placement succeeds without making any further queries.
Repeating this process until it succeeds therefore gives a simple randomized procedure.
We also note that we can remove the randomness using the standard method of conditional expectations \cite{raghavan1988,srivastav-stangier1996}, by placing these pieces one at a time while preserving the same guarantee.
Hence, we can find a valid placement deterministically using polynomially many arithmetic operations in the number of small pieces, so this step needs only modest computation between queries.

For the second question, our proof gives a clear answer.
We maintain only polynomially many partial allocations, make only polynomially many calls to \Core and \SubCore, and use only polynomially many queries outside these calls.
Thus, up to polynomial factors, the full protocol costs no more than one call to \Core or \SubCore involving $n$ agents.
In particular, combining the partial allocations does not add another exponential factor.

Hence, the exponential cost comes only from calls to \Core and \SubCore involving $\Theta(n)$ agents.
Naturally, this suggests two ways forward.
One is to improve the query bounds for these subroutines, and the other is to avoid calling them with so many agents.
Indeed, if every call involved $o(n)$ agents, then our protocol would use $2^{o(n)}$ queries.
Recently, Mackenzie and Suzuki \cite{mackenzie-suzuki} showed that we can avoid these subroutines altogether if we allow randomness.
By replacing \Core and \SubCore with randomized procedures, they obtain a protocol with polynomially many queries in expectation.
However, the deterministic case remains open.
Can either way forward yield a deterministic protocol with subexponentially or even polynomially many queries?
Or is the quadratic lower bound tight?
We believe that answering these questions will deepen our understanding of the computational limits of fair division, and we hope our approach provides a useful starting point.

\section*{Acknowledgments}

We thank Chaoran Yu for helpful feedback on a preliminary draft of this manuscript.

\bibliographystyle{plainnat}
\begingroup
\makeatletter
\def\hyper@linkurl#1#2{#1}
\makeatother
\bibliography{references}
\endgroup

\clearpage
\appendix
\section{AI Usage Disclosure}\label[appendix]{app:ai-acknowledgment}

We used ChatGPT 5.5 Pro/5.6 Sol extensively during the \emph{early} brainstorming stage, and Claude Opus 5.5 in polishing this version (v2). We briefly describe their roles here. GPT's most important contribution was the suggestion that the rather computationally heavy \texttt{GoLeft} protocol of Aziz and Mackenzie could be replaced by a much more efficient approach. It developed this suggestion into a complete protocol using $2^{2n(n-1)+O(n\log n)}$ Robertson--Webb queries. This GPT-generated protocol provided the conceptual backbone for our work and contained prototypes of several structural ideas used in our paper.

From this starting point, we refined and adapted GPT's structural ideas to the needs of our protocol. In our construction, these refined ideas eventually take the form of rejected labels, witnesses, profiles, and ranks, which we use to determine how snapshot pieces are chosen and redistributed.

Beyond this structural work, and independent of LLM assistance, we studied where the quadratic exponent in GPT's bound came from, and we found the source to be twofold.
First, GPT's construction kept exponentially many snapshots in order to find sufficiently many with exactly the same rejected labels (roughly speaking, this is the precondition of \Cref{lem:matching-selection}).
Second, it constructed the additions one at a time, and an unsuccessful attempt could require returning the cake and repeating this work across the stored snapshots.
Together, these choices made the cost grow like the square of the number of snapshots.

To address these inefficiencies, we introduced two further ideas.
First, we begin with only polynomially many snapshots and carefully prune them (\Cref{sec:reassignments-proof}). We show that the resulting collection still contains enough witnesses to reserve a different snapshot for every assignment we need (\Cref{lem:stable-pool,lem:matching-selection}).
Second, to avoid constructing the additions one at a time, we developed a ``simultaneous cutting'' lemma (\Cref{lem:simultaneous-cut}), which constructs all additions and rank pieces together using polynomially many queries (\Cref{lem:joint-pieces}).
These ideas further improve GPT's $2^{2n(n-1)+O(n\log n)}$ bound to $2^{O(n)}$ and allow us to streamline the algorithm (in particular, the recursion logic becomes much easier to describe).

While preparing the first version (v1), we used GPT 5.6 to help identify part of additional related work and to generate the pseudocode in \Cref{app:protocol-details,app:subcore}. In preparing this version (v2), we also used Opus 5.5 to add many illustrative diagrams, to simplify our v1 proofs, and to suggest improvements to the exposition.

Additionally, Opus 5.5 also found a gap in our v1.
There, our variant of \SubCore relied on the tie-breaking rule of Aziz and Mackenzie, and under one reading of that rule, our variant can create envy.
We repaired this gap with the help of Opus 5.5.
Since we want readers to see exactly which part of the paper this affects, we intentionally isolate the repair and its proofs in \Cref{app:subcore}.

Starting from GPT's backbone, and apart from \Cref{app:subcore}, we developed the protocol and proofs in this paper and wrote the remainder (and hence, almost all) of the paper, and we take full responsibility for its correctness and integrity.

\section{Additional Related Work}\label[appendix]{app:related-work}

Here we expand on the related work mentioned briefly in \Cref{sec:introduction}.
Recall that our theorem gives complete envy-free allocations for arbitrary nonatomic additive valuations, and that an agent's piece may be disconnected.
Since much of this work changes something in this setting, we group it by what it changes.

\paragraph{Connected pieces.}
Allowing disconnected pieces is essential to our theorem.
Indeed, if every agent must receive one interval, Stromquist proved that no finite protocol can guarantee a complete envy-free allocation for three or more agents \cite{stromquist2008}.
For additive valuations, work has therefore studied query bounds, approximation algorithms, and welfare optimization for connected pieces \cite{branzei-nisan2022,goldberg-hollender-suksompong2020,arunachaleswaran-et-al2019,barman-kulkarni2023}.
With free disposal, bounded envy-free protocols can instead give every agent a connected piece of positive value while allowing some cake to remain unallocated \cite{segal-halevi2015}.

\paragraph{More general utilities.}
One can also allow utilities more general than the additive valuations in our theorem.
In the model of continuous utilities over connected pieces, Deng, Qi, and Saberi gave an efficient approximation algorithm for three monotone agents \cite{deng-qi-saberi2012}.
Hollender and Rubinstein extended this positive result to four agents while proving hardness without monotonicity \cite{hollender-rubinstein2025}, and Hollender, Maystre, and Risse have since shown intractability for three agents with identical nonmonotone valuations \cite{hollender-maystre-risse2026}.
For a fixed number of possibly disconnected pieces, Gao et al. proved PPAD-hardness in another nonadditive model, even when the output need only make three agents approximately envy-free \cite{gao-et-al2024}.

\paragraph{Structured valuations.}
In the other direction, structural restrictions on additive valuations lead to stronger algorithms.
Kurokawa, Lai, and Procaccia gave a Robertson--Webb protocol using $O(n^6k\log k)$ queries when the agents' piecewise-linear densities have at most $k$ breakpoints in total \cite{kurokawa2013}, and Br\^anzei gave an envy-free protocol for polynomial value densities whose query bound also depends on the maximum degree \cite{branzei2015}.
In a direct-revelation model, Aziz and Ye gave polynomial algorithms for piecewise-constant valuations \cite{aziz-ye2014}.
For a restricted single-peaked model, Wang and Wu gave an exact envy-free protocol using a linear number of queries \cite{wang-wu2019}, while the monotone-likelihood-ratio condition yields efficient algorithms for connected allocations that are envy-free to arbitrary precision \cite{barman-rathi2022}.

\paragraph{Other fairness notions.}
Instead of changing the valuations, one can change what we ask of the allocation.
Abebe, Kleinberg, and Parkes introduced local envy-freeness on a graph \cite{abebe-kleinberg-parkes2017}.
For trees, Bei, Qiao, and Zhang gave a moving-knife algorithm \cite{bei-qiao-zhang2017}, and subsequent work gave a bounded discrete protocol \cite{ghalme-tree2024}.
Mehra and Psomas introduced hierarchies of fairness notions that connect proportionality, envy-freeness, and super envy-freeness and studied their query complexity \cite{mehra-psomas2025}.

\paragraph{Analyzing and verifying protocols.}
Finally, one can look more closely at existing protocols.
The bounded protocol for four agents that preceded AM's general protocol was later simplified, and its query bound was reduced by a factor of $3.4$ \cite{amanatidis2018}.
Beyond worst-case analysis, Ch\`eze proved that Webb's protocol uses fewer than $n^{12}$ queries with high probability under the uniform full-independence model \cite{cheze2020}.
The complexity of modern protocols has also prompted automated verification tools, including the \textsc{Slice} language, the \textsc{Crumbs} approach based on bounded model checking, and a faster \textsc{Slice} verification framework \cite{bertram-levinson-hsu2023,lester2024,bertram-lai-hsu2024}.

\section{Robertson--Webb Implementation Details}\label[appendix]{app:preliminaries-details}

In this appendix, we explain how ordinary Robertson--Webb queries handle pieces with several interval components and how we keep the protocol deterministic.

\subsection{Queries on Ordered Pieces}\label{app:rw-model-details}

We begin with two conventions.
We ignore overlaps at finitely many cut endpoints, which have value zero for every agent.
In Robertson--Webb query complexity, exact arithmetic, comparisons, and finite computation on oracle answers are free, and we do not bound the number of bits needed to represent arbitrary real answers.

Our protocol often asks an agent to cut a piece with several interval components, for instance when \SubCore trims a share.
Since a cut query measures out cake along a single interval, we first need to say what such a cut means.
To this end, we store the components of such a piece in a fixed order, for instance from left to right.
We call a piece with such an order an \emph{ordered piece}.
For disjoint pieces $P,A,M$, we write $[P,A,M]$ for their union, with the components of $P$ first, followed by those of $A$ and then those of $M$.
A \emph{prefix} ending at a point $x$ contains everything through $x$ in this order, while the corresponding \emph{suffix} contains everything from $x$ onward.
For example, take
\[
  P=[0,0.1]\cup[0.7,0.8],\qquad
  A=[0.2,0.3],\qquad
  M=[0.4,0.5]\cup[0.9,1].
\]
Then $[P,A,M]$ orders its components as
\[
  [0,0.1],\ [0.7,0.8],\ [0.2,0.3],\ [0.4,0.5],\ [0.9,1],
\]
and splitting the component $[0.2,0.3]$ at $x=0.25$ gives the prefix $P\cup[0.2,0.25]$ and the suffix $[0.25,0.3]\cup M$.
After trimming an ordered piece, we keep its remaining interval components in the same order.

\begin{proposition}\label{prop:ordered-piece-simulation}
Let $N$ be a set of agents.
Queries to an ordered piece can be implemented with ordinary Robertson--Webb queries.
Once we know how every agent in $N$ values every interval component, we can evaluate the ordered piece without further queries.
Asking for a prefix or suffix of a given value uses at most one ordinary cut query and $O(\lvert N\rvert)$ evaluation queries, regardless of the number of components.
\end{proposition}

\begin{proof}
The idea is to always keep track of these values for every agent in the current call.
At the top level, one evaluation query per agent records the value of the initial interval $[0,1]$.
If a later recursive call receives a subcake with several components, those components were cut earlier, so we already know how every agent values them.
The value of any ordered piece is therefore the sum of the known values of its components.

Now suppose agent $c$ asks for a prefix worth $a$.
Starting with the first component, we add $c$'s values until the sum reaches or exceeds $a$.
If it reaches $a$ at a component boundary, that boundary gives the desired prefix without a cut query.
Otherwise, the prefix ends inside the last component we added, where one ordinary cut query locates its endpoint.
A suffix is handled by cutting a prefix of the complementary value.

Finally, we need to keep our records up to date.
After a component is split, the amount requested in the cut query gives $c$'s value for one resulting interval, and subtraction gives her value for the other.
Every other agent in $N$ evaluates one of the two intervals once, after which subtraction again gives the value of the other.
We therefore know all new component values and can maintain the same information throughout the protocol.
\end{proof}

\subsection{Deterministic Choices}\label{app:deterministic-choices}

Since we want a deterministic protocol, we resolve every finite choice by a fixed rule.
We treat ties inside \SubCore in \Cref{app:subcore}.
If the cutter values the residue at zero while we are collecting snapshots, she values every piece equally, so we simply divide the residue into one possibly empty piece for each active agent by a fixed rule.

Ties can also arise when we cut an ordered piece, since an agent may value some components at zero, so that several component boundaries give the requested value.
In this case, we use the first such boundary in the order of the piece.
Thus, when an agent values the whole ordered piece at zero, this rule gives her the empty prefix and the full suffix.

\section{Omitted Proofs for Section 4}\label[appendix]{app:omitted-proofs-section-four}

\begin{proof}[Omitted proof of \Cref{lem:allocation-facts}]
For (i), each of $X$ and $Y$ gives every participating agent a nonnegative advantage over every other, so summing over the disjoint subcakes proves both claims.
For (ii), let $R$ and $R'$ be the residues of $X$ and $X\oplus Y$, respectively.
Since $i$ dominates $j$ under $X$, we have $\Delta_{ij}(X)\ge V_i(R)$.
Together with additivity, this gives
\[
  V_i(X_i\cup Y_i)-V_i(X_j\cup Y_j)
  \ge V_i(R)+V_i(Y_i)-V_i(Y_j)\ge V_i(R').
\]
The final inequality holds because $V_i(Y_i)\ge0$ and $R'\subseteq R\backslash Y_j$.
\end{proof}

\begin{proof}[Omitted proof of \Cref{cor:recursive-reductions}]
By \Cref{lem:allocation-facts} (ii), agents outside $E$ do not envy agents in $E$, even after the latter receive the residue.
Agents in $E$ cannot begin to envy an outside agent because only members of $E$ receive more cake.
Finally, by \Cref{lem:allocation-facts} (i), envy-freeness within $E$ is preserved, while comparisons within $N\backslash E$ remain unchanged.
\end{proof}

\begin{proof}[Omitted proof of \Cref{lem:simultaneous-cut}]
Let $m=\lvert I\rvert$.
As outlined in \Cref{sec:partial-allocation-proof}, we first want to divide $R$ into pieces that are small for every agent, where small means at most a fraction $\gamma=\delta^2/(2\log(2nm+2))$ of her value for $R$.
As in \Cref{prop:ordered-piece-simulation}, we know how every agent values every interval of $R$.

We ask each agent $i$ with $V_i(R)>0$ to cut $R$, following its fixed order (recall the convention from \Cref{app:rw-model-details}), into pieces worth $\gamma V_i(R)$ to her, except possibly the last piece, which may be smaller.
We combine the points where the agents cut $R$ with the endpoints of the intervals that make up $R$.
Together, these points divide $R$ into smaller pieces, which we call \emph{cells}.
Whenever an agent makes one of these cuts, we use \Cref{prop:ordered-piece-simulation} to record every agent's value for the two resulting intervals.
Consequently, once all cuts are made, we know every agent's value for every cell.
By construction, every cell $K$ satisfies $V_i(K)\le\gamma V_i(R)$ for every agent $i$.

We next show that we can form the desired pieces from these cells.
To this end, we place each cell independently at random, in $C_j$ with probability $\lambda_j$ for each $j\in I$, and leave it unused with the remaining probability $1-\sum_{j\in I}\lambda_j$.
Fix an agent $i$ with $V_i(R)>0$ and some $j\in I$, and let $x_K=V_i(K)/V_i(R)$, i.e., the fraction of $i$'s value for $R$ contained in cell $K$.
Because the cells partition $R$, we have $\sum_Kx_K=1$.
Moreover, $x_K\le\gamma$ for every cell, and hence
\[
  \sum_Kx_K^2\le\gamma\sum_Kx_K=\gamma.
\]
The fraction of $i$'s value placed in $C_j$ is the sum of $x_K$ over the cells placed in $C_j$, and its expected value is $\lambda_j$.
Therefore, Hoeffding's inequality gives
\[
\begin{aligned}
  \Pr\!\left[
    \left\lvert
      \sum_Kx_K\mathbf 1_{\{K\text{ is placed in }C_j\}}
      -\lambda_j
    \right\rvert>\delta
  \right]
  &\le2\exp\!\left(-\frac{2\delta^2}{\sum_Kx_K^2}\right)\\
  &\le2e^{-2\delta^2/\gamma}.
\end{aligned}
\]
Since there are at most $nm$ choices of $i$ and $j$, a union bound shows that the probability that any of the desired inequalities fails is at most
\[
  2nm\,e^{-2\delta^2/\gamma}
  =\frac{2nm}
    {(2nm+2)^4}
  <1.
\]
Thus, we can place the cells so that $\bigl\lvert V_i(C_j)-\lambda_jV_i(R)\bigr\rvert\le\delta V_i(R)$ for every $i\in N$ and $j\in I$, where the inequality is trivial if $V_i(R)=0$.
The pieces $C_j$ are pairwise disjoint because every cell is placed in at most one of them.

It remains to find one way to place the cells.
Since there are finitely many cells and we know every agent's value for each one, we can check all possibilities until we find one for which every inequality holds.
The argument above shows that one exists, and this search uses no further queries.

Finally, we count the queries used to create the cells and record their values.
Each agent needs $O(1/\gamma)$ Robertson--Webb cut queries to divide $R$ into her small pieces, so all agents together need $O(n/\gamma)$ such queries.
By \Cref{prop:ordered-piece-simulation}, recording every agent's value after each cut costs $O(n)$ evaluation queries.
Thus, the procedure uses $O(n^2/\gamma)=O\!\left(n^2\delta^{-2}\log(n\lvert I\rvert)\right)$ queries, as claimed.
\end{proof}

\begin{proof}[Omitted proof of \Cref{lem:joint-pieces}]
We first want to check that the sizes we ask for sum to less than $1$, so that \Cref{lem:simultaneous-cut} applies.
Since $\epsilon_i=LV_i(R^*)/[4n(Q+1)]$, every ratio $m_{s,i,k}/V_i(R^*)$ appearing in the maximum that defines $r_{s,k}$ is at most $L/[4n(Q+1)]$.
Thus, every addition has requested size at most
\[
  r_{s,k}+\delta
  \le\frac{5L}{16n(Q+1)}.
\]
There are at most $nQ$ additions.
By \Cref{lem:rank-properties}, the total requested size of all additions and all $P_h$ with $\ell_h > 0$ is at most
\[
\begin{aligned}
  \frac{5QL}{16(Q+1)}+L\sum_h\ell_h+n\delta
  &\le\frac{5QL}{16(Q+1)}+LQ+\frac{L}{16(Q+1)}\\
  &=LQ+\frac{(5Q+1)L}{16(Q+1)}\\
  &<1.
\end{aligned}
\]
Therefore, \Cref{lem:simultaneous-cut} applies, and it gives pairwise disjoint subsets of $R^*$.

For (i), suppose that $m_{s,i,k}\le\epsilon_i$.
The lower bound in \Cref{lem:simultaneous-cut} gives
\[
  V_i(A_{s,k})
  \ge r_{s,k}V_i(R^*)
  \ge m_{s,i,k}.
\]
For (ii), \Cref{lem:simultaneous-cut} gives, for every agent $i$,
\[
  V_i(A_{s,k})
  \le(r_{s,k}+2\delta)V_i(R^*)
  \le\frac{3L}{8n(Q+1)}V_i(R^*)
  \le\eta_i.
\]

For (iii), if $\ell_h>0$, we asked for $P_h$ to contain approximately the fraction $L(\ell_h-1/2)+\delta$ of $R^*$ for every agent.
Thus, \Cref{lem:simultaneous-cut} gives
\[
  L(\ell_h-1/2)V_i(R^*)
  \le V_i(P_h)
  \le\bigl[L(\ell_h-1/2)+2\delta\bigr]V_i(R^*)
  \le L\ell_hV_i(R^*),
\]
because $2\delta\le L/2$.
If $\ell_h=0$, then $P_h=\varnothing$, and (iii) holds trivially.
Finally, the number of pieces and $1/\delta$ are polynomial in $n$, so the query bound follows from \Cref{lem:simultaneous-cut}.
\end{proof}

\section{Full Pseudocode}\label[appendix]{app:protocol-details}

This appendix gives pseudocode for the protocol in \Cref{sec:main-theorem,sec:proof-roadmap}.
We use the base protocols from \Cref{sec:base-cases}.
We give \Core and \SubCore, which this pseudocode calls, as \Cref{alg:core-formal,alg:subcore-formal} in \Cref{app:subcore}.
Whenever the pseudocode asks us to choose among finitely many options, we use the rule in \Cref{app:deterministic-choices}.
For a collection $\mathcal T$ of snapshots, $\operatorname{Id}(\mathcal T)$ gives each agent $i$ the label-$i$ piece from every snapshot in $\mathcal T$.
As in \Cref{prop:core-guarantees} (iii), $\mathsf{Shrink}(N,U,K,c,D)$ repeatedly calls \Core on successive residues with agents $N$ and cutter $c$.
It makes at most $D$ calls, stops if $V_c(U)=0$, adds the resulting allocations to $K$, and returns the new allocation and residue.

\clearpage

\begin{algorithm}[H]
\scriptsize
\DontPrintSemicolon
\caption{$\Main(N,C)$}\label{alg:main-formal}
\KwIn{A nonempty active agent set $N$ and a subcake $C$ formed from finitely many intervals}
\KwOut{A complete envy-free allocation of $C$ among $N$}
Set $n\gets\lvert N\rvert$\AlgoLine
\If{$n\le3$}{
  \Return the allocation from \Cref{sec:base-cases}\AlgoLine
}
\tcp{stage 1: pre-processing}
Set $Q\gets n(n-1)$ and $(H,\mathcal S,R^*,m)\gets\mathsf{Preprocess}(N,C)$\AlgoLine
\tcp{stage 2: prune the snapshots}
\If{some $i\in N$ has $V_i(R^*)=0$}{
  Choose such an $i$ and \Return $\bigl(H\oplus\operatorname{Id}(\mathcal S)\bigr)\oplus\Main(N\backslash\{i\},R^*)$\AlgoLine
}
Set $\epsilon_i\gets V_i(R^*)/[4n^2(Q+1)Q^2]$ for every $i\in N$, and set $F_{s,i}\gets\{k\in N:m_{s,i,k}>\epsilon_i\}$ for every $s\in\mathcal S$ and $i\in N$\AlgoLine
Set $\mathcal P\gets\mathcal S$\AlgoLine
\While{some pair $(i,h)$ has between $1$ and $Q-1$ witnesses in $\mathcal P$}{
  Choose such a pair and remove all of its witness snapshots from $\mathcal P$\AlgoLine
}
For every $s\in\mathcal P$, form $G_s$ whose edges are $i\to k$ with $k\notin F_{s,i}$\AlgoLine
\If{some $G_s$ is not strongly connected}{
  Choose such an $s$ and a sink strongly connected component $S$ of $G_s$ as in \Cref{lem:scc}\AlgoLine
  Set $Y\gets H\oplus\operatorname{Id}(\mathcal S)$ and $U\gets R^*$\AlgoLine
  \ForEach{$i\in S$}{
    Set $(Y,U)\gets\mathsf{Shrink}(N,U,Y,i,\lceil\log_{n/(n-2)}Q^2\rceil)$\AlgoLine
  }
  \Return $Y\oplus\Main(N\backslash S,U)$\AlgoLine
}
\tcp{stage 3: construct the partial allocation}
Set $\mathbf F_i\gets(F_{s,i})_{s\in\mathcal P}$ and let $\mathcal E$ be the partition of $N$ into profile classes\AlgoLine
For every $i,h$, set $\mathcal W_{i,h}\gets\{s\in\mathcal P:F_{s,i}\nsubseteq F_{s,h}\}$\AlgoLine
Set $\mathcal T\gets\{(i,h):\mathbf F_i\npreceq\mathbf F_h\}$ and $\mathcal J\gets\varnothing$\AlgoLine
\tcp{reserve witnesses and choose matchings (\Cref{lem:matching-selection})}
\ForEach{$(i,h)\in\mathcal T$}{
  Choose $s\in\mathcal W_{i,h}\backslash\mathcal J$ and $k\in F_{s,i}\backslash F_{s,h}$\AlgoLine
  Choose a perfect matching $\mu_s$ in $G_s$ satisfying $\mu_s(h)=k$, and set $\mathcal J\gets\mathcal J\cup\{s\}$\AlgoLine
}
For every $E\in\mathcal E$, let $\ell_E$ be its rank as defined in \Cref{sec:partial-allocation-proof}, and set $\ell_h\gets\ell_E$ for every $h\in E$\AlgoLine
Set $(Z,R_1)\gets\mathsf{Allocate}(N,R^*;\mathcal J,\mu;\mathcal E,\ell)$\AlgoLine
Set $X\gets H\oplus\operatorname{Id}(\mathcal S\backslash\mathcal J)\oplus Z$\AlgoLine
Choose a profile class $E_*\subsetneq N$, which exists by \Cref{lem:positive-cutter}, and set $U\gets R_1$\AlgoLine
\ForEach{$i\in N\backslash E_*$}{
  Set $(X,U)\gets\mathsf{Shrink}(N,U,X,i,\lceil\log_{n/(n-2)}(4n(Q+1)Q^2)\rceil)$\AlgoLine
}
\Return $X\oplus\Main(E_*,U)$\AlgoLine
\end{algorithm}

\begin{algorithm}[H]
\scriptsize
\DontPrintSemicolon
\caption{$\mathsf{Preprocess}(N,C)$}\label{alg:preprocess-formal}
\KwIn{An active agent set $N$ and its subcake $C$}
\KwOut{An envy-free partial allocation $H$, snapshots $\mathcal S$, the final residue $R^*$, and the margins $m$}
Set $n\gets\lvert N\rvert$ and $Q\gets n(n-1)$\AlgoLine
Set $U\gets C$, $\mathcal S\gets\varnothing$, and $H$ to the empty allocation\AlgoLine
\tcp{construct the snapshots}
Fix an agent $c\in N$ as the cutter\AlgoLine
\For{$t\gets1$ \KwTo $Q^2$}{
  Set $(X,U)\gets\Core(N,U,c)$ and store $X$ as snapshot $s=t$\AlgoLine
  Write $c_{s,k}$ for its label-$k$ piece for every $k$, and set $\mathcal S\gets\mathcal S\cup\{s\}$\AlgoLine
}
Record $V_i(c_{s,k})$ for every $s,i,k$\AlgoLine
Set $m_{s,i,k}\gets V_i(c_{s,i})-V_i(c_{s,k})$ for every $s,i,k$\AlgoLine
\tcp{eliminate the intermediate margins}
Set $\Gamma_i\gets V_i(U)/Q^2$ and $\epsilon_i\gets\Gamma_i/[4n^2(Q+1)]$ for every $i$\AlgoLine
Set $D_{\mathrm{mid}}\gets\lceil\log_{n/(n-2)}(4n^2(Q+1))\rceil$\AlgoLine
\While{some agent $i$ has an intermediate margin}{
  Choose such an $i$ and set $(H,U)\gets\mathsf{Shrink}(N,U,H,i,D_{\mathrm{mid}})$\AlgoLine
  Recompute $\Gamma_r,\epsilon_r$ from the new $U$ for every $r\in N$\AlgoLine
}
Set $R^*\gets U$ and \Return $(H,\mathcal S,R^*,m)$\AlgoLine
\end{algorithm}

\clearpage

\begin{algorithm}[H]
\scriptsize
\DontPrintSemicolon
\caption{$\mathsf{Allocate}(N,R^*;\mathcal J,\mu;\mathcal E,\ell)$}\label{alg:allocate-formal}
\KwIn{The active set $N$ and residue $R^*$; the snapshots $\mathcal J$ and their matchings $\mu$; the profile classes $\mathcal E$ and their ranks $\ell$}
\KwOut{An allocation $Z$ and the residue $R_1$}
Set $n\gets\lvert N\rvert$, $Q\gets n(n-1)$, $L\gets1/(nQ^2)$, and $\delta\gets L/[16n(Q+1)]$\AlgoLine
Set $m_{s,i,k}\gets V_i(c_{s,i})-V_i(c_{s,k})$ for every $s\in\mathcal J$ and $i,k\in N$\AlgoLine
Set $\epsilon_i\gets V_i(R^*)/[4n^2(Q+1)Q^2]$ for every $i\in N$\AlgoLine
\tcp{cut the additions and rank pieces (\Cref{lem:joint-pieces})}
For every $s\in\mathcal J$ and $k\in N$, set
$r_{s,k}\gets\max\{0,\ m_{s,i,k}/V_i(R^*)\text{ where }i\in N,\ m_{s,i,k}\le\epsilon_i\}$\AlgoLine
Set $I\gets\{\mathsf{a}_{s,k}:(s,k)\in\mathcal J\times N\}\cup\{\mathsf{p}_h:\ell_h>0\}$\AlgoLine
Set $\lambda_{\mathsf{a}_{s,k}}\gets r_{s,k}+\delta$ for every $(s,k)\in\mathcal J\times N$ and $\lambda_{\mathsf{p}_h}\gets L(\ell_h-1/2)+\delta$ for every $h$ with $\ell_h>0$\AlgoLine
Apply \Cref{alg:cut-formal} to $N,R^*,I,(\lambda_j)_{j\in I},\delta$, and denote its output by $(C_j)_{j\in I}$\AlgoLine
Set $A_{s,k}\gets C_{\mathsf{a}_{s,k}}$ for every $(s,k)\in\mathcal J\times N$; set $P_h\gets C_{\mathsf{p}_h}$ when $\ell_h>0$ and $P_h\gets\varnothing$ otherwise\AlgoLine
Set $M_h\gets\bigcup_{s\in\mathcal J}c_{s,\mu_s(h)}$ and $A_h\gets\bigcup_{s\in\mathcal J}A_{s,\mu_s(h)}$ for every $h$, and set $B_i\gets\sum_{s\in\mathcal J}V_i(c_{s,i})$ for every $i$\AlgoLine
\tcp{redistribute the cake within each profile class (\Cref{lem:class-settlement})}
\ForEach{$E\in\mathcal E$}{
  Set $W_h\gets[P_h,A_h,M_h]$ for $h\in E$\AlgoLine
  Set up the lists of these pieces and $\varnothing$ as in \Cref{app:zero-value}, and set $Z^E\gets\SubCore(E,\{W_h:h\in E\}\cup\{\varnothing\},0)$\AlgoLine
  Record every agent's value for the interval components created by each trim\AlgoLine
  For every $h\in E$, let $Y_h$ be the unused prefix of $W_h$\AlgoLine
}
Set $Z\gets\bigoplus_{E\in\mathcal E}Z^E$\AlgoLine
Set $R_1\gets\bigl(R^*\backslash\bigcup_{h\in N}(P_h\cup A_h)\bigr)\cup\bigcup_{h\in N}Y_h$\AlgoLine
\Return $(Z,R_1)$\AlgoLine
\end{algorithm}

\begin{algorithm}[H]
\scriptsize
\DontPrintSemicolon
\caption{Procedure in \Cref{lem:simultaneous-cut}}\label{alg:cut-formal}
\KwIn{An agent set $N$; an ordered piece $R$, together with every agent's value for each interval that forms $R$; a nonempty finite index set $I$; numbers $\lambda_j\ge0$ satisfying $\sum_{j\in I}\lambda_j\le1$; and $0<\delta<1$}
\KwOut{Pairwise disjoint pieces $(C_j)_{j\in I}$ satisfying $\lvert V_i(C_j)-\lambda_jV_i(R)\rvert\le\delta V_i(R)$ for every $i\in N,j\in I$}
Set $n\gets\lvert N\rvert$, $m\gets\lvert I\rvert$, and $\gamma\gets\delta^2/(2\log(2nm+2))$\AlgoLine
\ForEach{$i\in N$ with $V_i(R)>0$}{
  Following the order of $R$, ask agent $i$ for a prefix worth $k\gamma V_i(R)$ for every positive integer $k$ with $k\gamma<1$, and record its endpoint\AlgoLine
  After each cut, record every agent's value for the two resulting interval components as in \Cref{prop:ordered-piece-simulation}\AlgoLine
}
Together, these cut points and the endpoints of the intervals forming $R$ divide $R$ into cells. Let $\mathcal K$ be their collection, and use the recorded values to compute how every agent values every cell\AlgoLine
\ForEach{labeling $\tau$ that assigns each cell in $\mathcal K$ an index in $I$ or leaves it unused}{
  Set $C_j(\tau)\gets\bigcup\{K\in\mathcal K:\tau(K)=j\}$ for every $j\in I$\AlgoLine
  \If{$\lvert V_i(C_j(\tau))-\lambda_jV_i(R)\rvert\le\delta V_i(R)$ for every $i\in N$ and $j\in I$}{
    \Return $(C_j(\tau))_{j\in I}$\AlgoLine
  }
}
\end{algorithm}

\section{\texorpdfstring{Our Variant of \SubCore}{Our Variant of SubCore}}\label[appendix]{app:subcore}
\addtocontents{toc}{\protect\setcounter{tocdepth}{1}}

In this appendix, we describe our variant of \SubCore and prove the results about \Core and \SubCore that we stated in \Cref{sec:core-subroutines-details}.
We want to keep all of this in one place, so that it can be read on its own.

We borrow three terms from Aziz and Mackenzie \cite{aziz-mackenzie}.
An agent's \emph{benchmark} is the value that she must receive at a given point, the \emph{contested} pieces are the ones tentatively allocated so far, and an agent \emph{trims} a piece by cutting it down to a suffix worth a given value to her.
Apart from these terms, we use our own notation.

\subsection{What We Change}\label{app:subcore-changes}

We restate \Core and \SubCore of Aziz and Mackenzie in our notation as \Cref{alg:core-formal,alg:subcore-formal}, and we change their \SubCore in two ways.

\begin{algorithm}[H]
\scriptsize
\DontPrintSemicolon
\caption{$\Core(N,R,c)$, due to Aziz and Mackenzie}\label{alg:core-formal}
\KwIn{An agent set $N$, a residue $R$, and a cutter $c\in N$}
\KwOut{An envy-free partial allocation of $R$ and its residue}
Set $n\gets\lvert N\rvert$, and ask $c$ to cut $R$ into a collection $\mathcal C$ of $n$ pieces that she values equally\AlgoLine
Run $\SubCore(N\backslash\{c\},\mathcal C,0)$, with benchmark $0$ for every agent\AlgoLine
Give $c$ a piece of $\mathcal C$ that this call leaves unallocated\AlgoLine
\Return the resulting allocation and the unallocated part of $R$\AlgoLine
\end{algorithm}

\begin{algorithm}[H]
\scriptsize
\DontPrintSemicolon
\caption{Our variant of $\SubCore(E,\mathcal C,b)$}\label{alg:subcore-formal}
\KwIn{A set $E$ of $m$ agents, numbered $1,\ldots,m$; a collection $\mathcal C$ of more than $m$ disjoint ordered pieces; a benchmark $b_j$ for every $j\in E$; and an envy-free allocation $H$ that gives every $j$ a suffix of a distinct piece worth exactly $b_j$ (when every benchmark is $0$, we take these suffixes to be empty)}
\KwOut{An envy-free allocation in which every $j\in E$ receives a suffix of a distinct piece worth at least $b_j$, at least one agent receives a whole piece, and no agent prefers an unallocated piece to her own}
\tcp{we keep AM's terms, such as benchmark and contested; our $T$ and $t$ are their $W$ and $m$}
\For{$t\gets1$ \KwTo $m$}{
  \eIf{agent $t$'s favorite piece of $\mathcal C$ is not tentatively allocated}{
    Tentatively give agent $t$ this piece\AlgoLine
  }{
    Call the $t-1$ tentatively allocated pieces \emph{contested}, and the other pieces of $\mathcal C$ \emph{uncontested}, and let $X$ be the tentative allocation\AlgoLine
    For every $j\le t$, set $b'_j$ to the maximum of $b_j$ and her value for her favorite uncontested piece\AlgoLine
    Every $j\le t$ trims each contested piece worth at least $b'_j$ to her down to a suffix worth exactly $b'_j$\AlgoLine
    On each contested piece, choose one agent with the rightmost trim, favoring the agent who holds a suffix of this piece in $H$ if she ties, and let $T$ be the set of chosen agents\AlgoLine
    Replace the tentative allocation by giving each agent in $T$ the suffix she leaves on a piece where we chose her\AlgoLine
    \While{$\lvert T\rvert<t-1$}{
      Ignore the trims of the agents in $T$ from now on\AlgoLine
      Let $\mathcal C'$ consist of the suffix of each contested piece left by its rightmost remaining trim; a contested piece with no remaining trim stays whole\AlgoLine
      Replace the tentative allocation by the result of $\SubCore(T,\mathcal C',(b'_j)_{j\in T})$, with the current tentative allocation as its $H$\AlgoLine
      If no unallocated piece of $\mathcal C'$ has a remaining trim, move agents in $T$ along a chain as in \Cref{lem:chain}\AlgoLine
      Choose an unallocated piece of $\mathcal C'$ that has a remaining trim, and let $i\le t$ be an agent with its rightmost remaining trim\AlgoLine
      Tentatively give $i$ this piece, and add $i$ to $T$\AlgoLine
      For every $j\in T$, set $b'_j$ to her value for her tentative piece\AlgoLine
    }
    \tcp{we skip line 12 of AM's Algorithm 2: ignore the trims of the agents in $T$, form $\mathcal C'$ as above, and run $\SubCore(T,\mathcal C',(b'_j)_{j\in T})$}
    Tentatively give the sole agent outside $T$ among the first $t$ her favorite uncontested piece\AlgoLine
  }
}
\Return the tentative allocation\AlgoLine
\end{algorithm}

First, we want to improve the query bound, so we skip line 12 of AM's Algorithm 2.
Once $\lvert T\rvert=t-1$, AM makes one final recursive call on $T$, while we omit that call.
Without this call, we need to say ourselves who receives what.
If $\lvert T\rvert=t-1$ as soon as the trims are made, then every contested piece has a different rightmost trimmer.
For each contested piece, we give its rightmost trimmer the part beginning at her trim.
If the while loop runs, we keep the tentative allocation obtained in its final iteration.
In either case, we give the sole agent outside $T$ her favorite uncontested piece.

Second, we want to handle ties ourselves.
Trims create exact ties on purpose, and these ties decide who may stay on a contested piece.
To break these ties, we want to know which piece each agent needs.
Therefore, we add to the input of \SubCore an envy-free allocation $H$ in which every agent's suffix is worth exactly her benchmark.
When several agents tie for the rightmost trim on a contested piece, we favor the agent who holds a suffix of that piece in $H$.
Moreover, if no unallocated piece of $\mathcal C'$ has a remaining trim, we move agents in $T$ along a chain of suffixes that they value equally, which frees a piece that has one.

\subsection{Correctness of Our Variant}\label{app:subcore-correctness}

We now want to show that \Cref{alg:subcore-formal} does what we promised.
To keep the argument simple, we assume in this section that every agent values every interval positively, so that she strictly prefers a longer suffix of the same piece.
We remove this assumption in \Cref{app:zero-value}.

\begin{proposition}\label{prop:subcore-correct}
Given an input as in \Cref{alg:subcore-formal}, the algorithm returns the output it promises.
Moreover, every allocated suffix is a whole piece or the suffix left by some agent's trim.
\end{proposition}

\begin{proof}[Proof of \Cref{prop:subcore-correct}]
We argue by induction on $m$, using the notation of \Cref{alg:subcore-formal}, and we let $Y$ be the allocation of $T$ returned by a recursive call.
Going through the for loop, we want the tentative allocation of the first $t$ agents to have all the promised properties after iteration $t$.
If agent $t$'s favorite piece is unallocated, she receives it whole.
Since her suffix in $H$ is worth $b_t$ to her, this piece is worth at least $b_t$.
She envies nobody, and nobody envies her, since nobody prefers an unallocated piece to her own.

Otherwise, agent $t$'s favorite piece is contested, and we need to make room for her.
Before we do so, we want to collect two facts about the trims.
Since $X$ meets every benchmark and nobody prefers an uncontested piece to her own, every $j<t$ values $X_j$ at least $b'_j$.
Therefore, she trims the piece of $X_j$, and her trim leaves a suffix contained in $X_j$.
In particular, every contested piece has a trim.
We will also use that an agent who makes no trim on a piece values all of it below $b'_j$.

We first want every agent who needs a contested piece to be in $T$, so that she never has to settle for an uncontested one.
Suppose $b_i$ exceeds $i$'s value for every uncontested piece, so that $b'_i=b_i$.
Then, $H_i$ is a suffix of a contested piece, and $i$ trims this piece exactly down to $H_i$.
Every other $j\le t$ values $H_i$ at most $V_j(H_j)=b_j\le b'_j$, so her trim, if she makes one, leaves a suffix containing $H_i$.
Hence, $i$ ties for the rightmost trim on this piece, which is exactly when we favor her, so $i\in T$.
Since agents never leave $T$, the sole agent outside $T$ at the end values her favorite uncontested piece at least her benchmark.

Next, we give each agent in $T$ the suffix she leaves on a piece where we chose her, and we want this allocation to be envy-free.
She values this suffix exactly $b'_j$, and any other agent $j\le t$ values it at most $b'_j$, since her trim, if she makes one, leaves a suffix containing it.
Therefore, this allocation is envy-free.
If $\lvert T\rvert=t-1$, this is the allocation we use.
Otherwise, its suffixes lie inside the pieces of $\mathcal C'$, so it can serve as $H$ for the first recursive call.

Now consider an iteration of the while loop, where we want to add one agent to $T$.
By the induction hypothesis, $Y$ is envy-free, it gives every $j\in T$ a suffix of a distinct piece of $\mathcal C'$ worth at least $b'_j$, and no agent in $T$ prefers an unallocated piece of $\mathcal C'$ to her own.
Note that the agents outside $T$ have not changed their $b'_j$, and each such $i$ values every piece of $\mathcal C'$ at most $b'_i$.
To grow $T$, we want an unallocated piece of $\mathcal C'$ with a remaining trim, since an agent $i$ with its rightmost remaining trim is then happy to take it.
Since $\lvert T\rvert<t-1$, some piece of $\mathcal C'$ is unallocated.
If it has no remaining trim, we need a small repair, which we postpone to \Cref{lem:chain}, and this repair gives us a piece that does.
We give this piece to $i$, who values it exactly $b'_i$.
Nobody in $T$ prefers it to her own suffix, and $i$ values every suffix of $Y$ at most $b'_i$, so the allocation of $T\cup\{i\}$ is envy-free.
Since removing $i$'s trims only enlarges the pieces of $\mathcal C'$, this allocation can serve as $H$ for the next recursive call, with $b'_j$ set to each agent's current value, which never decreases.

It remains to place the sole agent $x$ outside $T$.
When the loop ends, or if it never runs, $\lvert T\rvert=t-1$, so all contested pieces are allocated.
Every agent in $T$ values her suffix at least $b'_j$, and hence at least $b_j$ and at least every uncontested piece.
Agent $x$ receives her favorite uncontested piece, which we showed meets her benchmark.
This piece is worth exactly $b'_x$, and she values every contested suffix at most $b'_x$.
Thus, the allocation of the first $t$ agents is envy-free and meets every benchmark, one agent holds a whole piece, and nobody prefers an unallocated piece to her own.
Finally, each allocated suffix is a whole piece or the suffix left by some agent's trim, by construction and by the induction hypothesis.
This completes the induction.
\end{proof}

\begin{lemma}\label{lem:chain}
Consider an iteration of the while loop in which an unallocated piece $P$ of $\mathcal C'$ has no remaining trim.
Then, we can move agents in $T$ along a chain, so that every agent keeps her value, $P$ becomes allocated, and a piece of $\mathcal C'$ with a remaining trim becomes unallocated instead.
\end{lemma}

\begin{proof}
Since no agent outside $T$ wants $P$, we need an agent in $T$ to take it.
To find out who, we want to compare $Y$ with $X$.
Recall from the proof of \Cref{prop:subcore-correct} that $X$ and $Y$ are envy-free, that every agent outside $T$ who holds a suffix in $X$ values it at least $b'_j$, and that no agent in $T$ prefers an unallocated piece of $\mathcal C'$ to her own.
Let $i_0$ be the agent who holds $P$ under $X$, let $i_{r+1}$ be the agent who holds under $X$ the piece that $i_r$ receives under $Y$, and we stop at the first agent $i_k$ outside $T$.

We first want to make sure that this chain is well defined.
Since $P$ has no remaining trim, and every agent who holds a piece under $X$ trims it, we have $i_0\in T$.
The agents $i_r$ are distinct, since $X$ and $Y$ each give distinct agents distinct pieces and $P$ is unallocated under $Y$.
Thus, the chain stops.

We now walk back from the end of the chain, and we claim that $Y_{i_{r-1}}\subseteq X_{i_r}$ for every $1\le r\le k$.
For $r=k$, since $i_k$ values $X_{i_k}$ at least $b'_{i_k}$, her trim leaves a suffix contained in $X_{i_k}$.
The piece of $\mathcal C'$ starts at or after this trim, so it is contained in $X_{i_k}$, and so is $Y_{i_{k-1}}$.
Next, suppose the claim holds for $r+1$, which gives
\[
  V_{i_r}(Y_{i_r})\le V_{i_r}(X_{i_{r+1}})\le V_{i_r}(X_{i_r}),
\]
where the last inequality holds since $X$ is envy-free.
If $Y_{i_{r-1}}$ strictly contained $X_{i_r}$, then $i_r$ would value it strictly more than $X_{i_r}$, and hence more than her own $Y_{i_r}$, contradicting the envy-freeness of $Y$.
This proves the claim, and these inequalities hold for every $r<k$.

Now we walk forward.
Since $i_0$ does not prefer the unallocated $P$ to her own suffix, we have $V_{i_0}(Y_{i_0})\ge V_{i_0}(P)\ge V_{i_0}(X_{i_0})$.
Together with these inequalities for $r=0$, this forces equality throughout, so $X_{i_0}=P$ and $Y_{i_0}=X_{i_1}$.
Luckily, this passes along the chain.
Suppose $Y_{i_{r-1}}=X_{i_r}$ for some $1\le r<k$.
Since $Y$ is envy-free, we have $V_{i_r}(Y_{i_r})\ge V_{i_r}(Y_{i_{r-1}})=V_{i_r}(X_{i_r})$.
Together with the same inequalities, this again forces equality, so $i_r$ values $X_{i_r}$ exactly as much as $Y_{i_r}$, and $Y_{i_r}=X_{i_{r+1}}$.
In the end, $Y_{i_{k-1}}=X_{i_k}$.
Since $Y_{i_{k-1}}$ is contained in its piece of $\mathcal C'$, which is contained in the suffix left by $i_k$'s trim, which is contained in $X_{i_k}$, these all coincide.
Thus, $X_{i_k}$ is a piece of $\mathcal C'$ with a remaining trim.

We can now move the agents, as promised.
We move each $i_r$ with $r<k$ to $X_{i_r}$, and by the equalities above, every agent keeps her value.
Then, $i_0$ receives $P$, and each later $i_r$ receives the suffix that $i_{r-1}$ held.
Hence, the only change to the allocated suffixes is that $P$ is added and $X_{i_k}$ is removed, so the allocation stays envy-free, and nobody in $T$ prefers an unallocated piece of $\mathcal C'$ to her own.
\end{proof}

\subsection{\texorpdfstring{Breaking Ties in \SubCore}{Breaking Ties in SubCore}}\label{app:zero-value}

We now remove the assumption that every agent values every interval positively.
An agent may value a stretch of a piece at zero, so a longer suffix may not be strictly better for her, and our proof above needs this.
To fix this, we break such ties formally.
For each piece, we keep an ordered list of its suffixes, from longest to shortest, and we allow two suffixes in the list to start at the same point.
We also give each suffix in the list a rational \emph{tiebreaker} for each agent, whose only job is to break ties.
In \SubCore, we read an agent's value for such a suffix as the pair of her real value and her tiebreaker, compared lexicographically, and we keep these pairs strictly decreasing along the list.
Our proof above uses only that values are totally ordered, that they strictly decrease along a piece, and that an agent can trim to an exact value, so it holds for these pairs.
An allocation that is envy-free for the pairs is also envy-free for the real values, since these are the first coordinates.

Next, we need to say how these lists start.
We want $H$ to give every agent the empty suffix at the top level.
Therefore, we start the list of every input piece of a top-level call with the whole piece, with tiebreaker $1$, and its empty suffix, with tiebreaker $0$.
We take a benchmark of $0$ to mean the pair $(0,0)$, and we keep the list of a piece in recursive calls.
In the calls of \Cref{prop:subcore-required-minima}, we also want to protect $M_h$ and to avoid the empty piece.
Therefore, we add $M_h$ to the list of $W_h$ with tiebreaker $1/2$, and we use tiebreaker $1/4$ on the empty piece as a whole.

When agent $j$ trims a piece to a value $(\varphi,\theta)$, we want a suffix in the list that she values exactly that much.
If one exists, we use it.
Otherwise, we add one between the two neighboring suffixes that she values more and less, starting at a point between theirs where her real value is $\varphi$.
We take the start of the first neighbor if it works, then the start of the second if it works, and otherwise we make a cut query.
We give the new suffix the tiebreaker $\theta$ for her.
For each other agent, we choose a tiebreaker that puts her value strictly between her values for the two neighbors, which is possible since the rationals are dense.
By \Cref{prop:ordered-piece-simulation}, we already know every value we need, so this uses no queries beyond the cut query and its recording evaluations, as before.

Finally, these tiebreakers protect $M_h$, as we promised in \Cref{lem:class-settlement}.
By \Cref{prop:subcore-required-minima} and that lemma, every trim by an agent $i\in E$ is worth at least $(q_i,1)$ to her, and we also have $V_i(M_h)\le q_i$ for every $h\in E$.
Hence, this trim is worth more to her than the pair $(V_i(M_h),1/2)$, which is her value for $M_h$ in the list.
Thus, every trim leaves a suffix before $M_h$ in the list, so it never cuts into $M_h$.

\subsection{\texorpdfstring{Omitted Proofs for \Core and \SubCore}{Omitted Proofs for Core and SubCore}}\label{app:subcore-proofs}

\begin{proof}[Omitted proof of \Cref{prop:core-guarantees}]
For (i), recall that the cutter cuts the cake into $n$ pieces she values equally.
By \Cref{prop:subcore-correct}, the other agents receive suffixes of distinct pieces without envy, one of them receives a whole piece, and none of them prefers an unallocated piece to her own.
Therefore, nobody envies the unallocated piece we give to the cutter.
Since she values all $n$ pieces equally, she envies nobody either.
Thus, the cutter and the agent with the whole piece each receive a piece worth $V_c(R)/n$ to the cutter.
Part (ii) then follows immediately from (i), because two pieces worth $V_c(R)/n$ to $c$ are allocated.
Applying this repeatedly yields (iii), as
\[
  V_c(R')
  \le
  \left(\frac{n-2}{n}\right)^{
    \lceil\log_{n/(n-2)}\beta\rceil
  }V_c(R)
  \le\frac{V_c(R)}{\beta}.
\]
(If the sequence stops early, then $V_c(R')=0$, so the same bound is immediate.)
\end{proof}

\begin{proof}[Omitted proof of \Cref{prop:subroutine-query-bounds}]
By \Cref{prop:subcore-correct}, our variant satisfies the guarantees of \SubCore, which we list as the output of \Cref{alg:subcore-formal}.

We now count the queries.
Set $S_0=S_1=0$, and for $m\ge2$ define
\[
  S_m=S_{m-1}+m(m-1)+\sum_{j=1}^{m-2}S_j.
\]
We claim that \SubCore on $m$ agents makes at most $S_m$ cut queries.
Indeed, processing the first $m-1$ agents uses at most $S_{m-1}$ cuts, and introducing the $m$th agent requires at most $m(m-1)$ trims.
The recursive calls inside the while loop have distinct sizes in $\{1,\ldots,m-2\}$.
Moving agents along a chain makes no queries, since we already know how every agent values every suffix involved.
Thus, $S_m$ bounds the number of cuts by induction.
For $m\ge3$, subtracting the recurrence for $S_{m-1}$ gives
\[
  S_m=2S_{m-1}+2(m-1).
\]
The same identity holds for $m=2$, and $S_1=0$, so
\[
  S_m=2^{m+1}-2m-2=O(2^m).
\]
(This is where our change helps. AM's final call on $T$ would add another $S_{m-1}$, and the resulting recurrence grows with base $(3+\sqrt5)/2$.)
By \Cref{prop:ordered-piece-simulation}, each cut uses $O(m)$ evaluation queries to record how the agents value the two resulting parts.
Thus, the call uses $O(mS_m)=O(m2^m)$ evaluation queries.

A \Core call first uses $n-1$ cuts to create $n$ pieces that the cutter values equally, followed by $O(n^2)$ evaluations to record how the agents value them.
It then invokes \SubCore on $n-1$ agents, so the stated \Core bounds follow.
\end{proof}

\begin{proof}[Omitted proof of \Cref{prop:subcore-required-minima}]
By \Cref{prop:subcore-correct}, the call returns an envy-free allocation in which the agents receive suffixes of distinct input pieces.
Fix $i\in E$. We want to show that every suffix she trims a piece down to, as well as the one she receives, is worth at least $\min_{h\in E}V_i(W_h)$ to her.
We use the notation of \Cref{alg:subcore-formal}.
In particular, $b_j$ is the value that the input requires agent $j$ to receive, and $b'_j$ is the value that she must receive at a given point.
In this call, we set $b_j=0$ for every $j\in E$.

The key is that, whenever $i$ trims, some whole $W_h$ is still available, which keeps $b'_i$ at least this minimum.
To see this, consider an iteration of the for loop, no earlier than the one that processes $i$, in which the favorite piece of the agent being processed is already tentatively allocated.
Since at most $\lvert E\rvert-1$ of the $\lvert E\rvert+1$ pieces are tentatively allocated, at least two pieces are uncontested.
Every trim is made on a piece that is contested at that time, and as shown in the proof of \Cref{prop:subcore-correct}, every iteration ends with each contested piece tentatively allocated.
Hence, the uncontested pieces are untrimmed, and since only one input piece is empty, at least one of them is a whole $W_h$.
Since $b_i=0$, \Cref{alg:subcore-formal} then sets $b'_i$ to her value for her favorite uncontested piece, which is at least this minimum.

We claim that $b'_i$ does not decrease for the rest of this iteration.
Indeed, the while loop sets $b'_i$ to her value for her tentative piece, which is at least $b'_i$ by the proof of \Cref{prop:subcore-correct}.
Also, each recursive call in the while loop takes $b'_i$ as its input $b_i$, and since that call sets its own $b'_i$ to a maximum with $b_i$, the same argument shows that $b'_i$ never falls below this input inside that call.
Since every trim by an agent $j$ leaves a suffix worth exactly $b'_j$ to her, every suffix that $i$ trims a piece down to during this iteration is worth at least this minimum to her.
Moreover, by the proof of \Cref{prop:subcore-correct}, $i$ ends this iteration with a piece worth at least $b'_i$, and hence at least this minimum.

Since $i$ trims only during iterations of this kind, every suffix she trims a piece down to is worth at least this minimum to her.
In every other iteration after the one that processes $i$, the agent being processed receives an unallocated piece, so the piece that $i$ holds does not change.
Therefore, if an iteration of this kind exists, the suffix she receives is the piece she holds at the end of the last one.
Otherwise, she receives her favorite piece whole when she is processed and keeps it.
In both cases, the suffix she receives is worth at least this minimum.
By \Cref{prop:subcore-correct}, each allocated part is a whole piece or the suffix left by some agent's trim.
Since the pieces $W_h$ may be disconnected, we implement the corresponding trims using \Cref{prop:ordered-piece-simulation}.

It remains to show that the empty piece stays unallocated.
Indeed, whenever we assign a piece directly, some whole $W_h$ remains available, and by our tiebreakers in \Cref{app:zero-value}, every agent strictly prefers it to the empty piece.
Thus, we never assign the empty piece directly.
Since the recursive calls use only contested pieces, the empty piece never enters a recursive call either.
Therefore, it remains unallocated.
\end{proof}

\end{document}